\documentclass[letterpaper,11pt]{article}

\usepackage{amssymb,amsmath,amsthm,mathrsfs,bbm}
\usepackage[utf8]{inputenc} 
\usepackage[T1]{fontenc}    
\usepackage{url}   
\usepackage{cite}          
\usepackage{booktabs}       
\usepackage{amsfonts}       
\usepackage{nicefrac}       
\usepackage{microtype}      
\usepackage{etoolbox}
\usepackage{graphicx}
\usepackage{psfrag}
\usepackage{caption}
\usepackage{appendix}

\usepackage{hyperref}
\hypersetup{
    colorlinks=true,
    linkcolor=blue,
    filecolor=magenta,
    urlcolor=cyan,
}

\usepackage[usenames,dvipsnames]{xcolor}

\usepackage{amssymb,amsmath,amsthm,mathrsfs,mathtools}
\usepackage{dsfont}
\usepackage{amsbsy}
\usepackage{amsfonts}
\usepackage{graphicx}
\usepackage{colonequals}
\usepackage{psfrag}
\usepackage{mdframed}
\usepackage{enumitem}
\usepackage{wrapfig}
\usepackage{subcaption}
\usepackage{setspace}
\usepackage{booktabs}
\usepackage{multirow}
\usepackage{bm}

\allowdisplaybreaks

\newtheorem{assumptions}{Assumption} 

\makeatletter
\newtheorem*{rep@theorem}{\rep@title}
\newcommand{\newreptheorem}[2]{%
\newenvironment{rep#1}[1]{%
 \def\rep@title{#2 \ref{##1}}%
 \begin{rep@theorem}}%
 {\end{rep@theorem}}}
\makeatother
\newreptheorem{conj}{Conjecture}

\theoremstyle{definition}
 
\newtheorem{defn}{Definition} 
 
\newtheorem{remark}{Remark} 
 
\newtheoremstyle{tasksty}
  {3pt}
  {3pt}
  {\em}
  {3pt}
  {\color{MidnightBlue}\sffamily\mdseries\upshape}
  {:}
  { }
  {}

\theoremstyle{tasksty}

\newtheoremstyle{sGoal}
  {3pt}
  {3pt}
  {\em}
  {3pt}
  {\color{BrickRed}\sffamily\mdseries\upshape}
  {:}
  { }
  {}

\theoremstyle{sGoal}

\theoremstyle{plain}
\newtheorem{theorem}{Theorem}
\newtheorem{lemma}{Lemma}

\newcommand{\scP}{\mathscr{P}}
\newcommand{\bbR}{\mathbb{R}}
\newcommand{\be}{\begin{equation}}
\newcommand{\ee}{\end{equation}}

\newcommand{\KL}{\textup{KL}}

\newcommand{\sP}{\mathscr{P}}
\newcommand{\sR}{\mathscr{R}}
\newcommand{\sB}{\mathscr{B}}
\newcommand{\sC}{\mathscr{C}}

\newcommand{\BE}{\mathbb{E}}
\newcommand{\BN}{\mathbb{N}}
\newcommand{\BZ}{\mathbb{Z}}

\newcommand{\BR}{\mathbb{R}}

\newcommand{\calB}{\mathcal{B}}

\newcommand{\calY}{\mathcal{Y}}

\newcommand{\calJ}{\mathcal{J}}

\newcommand{\calF}{\mathcal{F}}

\newcommand{\supp}{\mathsf{supp}}

\newcommand{\Reals}{\mathbb{R}}

\newcommand{\bp}{\boldsymbol{p}}

\newcommand{\eps}{\varepsilon}

\definecolor{light-gray}{gray}{.90}
\definecolor{aliceblue}{rgb}{0.94, 0.97, 1.0}
\definecolor{airforceblue}{rgb}{0.36, 0.54, 0.66}
\definecolor{bleudefrance}{rgb}{0.19, 0.55, 0.91}
\definecolor{cerulean}{rgb}{0.0, 0.48, 0.65}

\newmdenv[%
  linewidth =0pt,%
  fontcolor=bleudefrance,
  innertopmargin=2pt,
  innertopmargin=2pt,
  leftmargin = 0pt,
  rightmargin = 0pt,
  innerleftmargin = 2pt,
  innerrightmargin = 2pt,
  skipabove = 6pt,%
  skipbelow = 6pt
]{RQ}

\newmdenv[%
  backgroundcolor=aliceblue, %
  linewidth = 2pt,%
  skipabove = 10pt,%
  skipbelow = 10pt,
  pstrickssetting={linestyle=dashed,},
  linecolor=airforceblue,
  middlelinewidth=2pt
]{TODO}

\usepackage[margin=1in]{geometry}
\usepackage{epstopdf}

\title{Cactus Mechanisms: Optimal Differential Privacy Mechanisms in the Large-Composition Regime}

\usepackage{lipsum}

\makeatletter
\def\blfootnote{\gdef\@thefnmark{}\@footnotetext}
\makeatother

\author{
Wael Alghamdi\thanks{Corresponding author, remaining authors in alphabetical order.}~\thanks{W.\hspace{-0.5pt} Alghamdi and F.\hspace{-0.5pt} P.\hspace{-0.5pt} Calmon are with the School of Engineering and Applied Science, Harvard University (emails: {alghamdi@g.harvard.edu, flavio@seas.harvard.edu})}~,
Shahab Asoodeh\thanks{S.\hspace{-0.5pt} Asoodeh is with the Department of Computing and Software,\hspace{-1pt} McMaster\hspace{-0.5pt} University\hspace{-1.5pt} (email:\hspace{-3pt} {asoodehs@mcmaster.ca})}~,
Flavio P. Calmon${}^\dagger$,\\
Oliver Kosut\thanks{O.\hspace{-0.5pt} Kosut, L.\hspace{-0.5pt} Sankar, and F.\hspace{-0.5pt} Wei are with the School of Electrical, Computer, and Energy Engineering, Arizona State University (emails: {\{\hbox{okosut},lsankar,fwei16\}@asu.edu})}~,
Lalitha Sankar${}^\S$, and
Fei Wei${}^\S$
}

\date{}

\begin{document}

\onecolumn

\maketitle

\begin{abstract}
Most differential privacy mechanisms are applied (i.e., composed) numerous times on sensitive data. We study the design of optimal differential privacy mechanisms in the limit of a large number of compositions. As a consequence of the law of large numbers, in this regime the best privacy mechanism is the one that minimizes the Kullback-Leibler divergence between the conditional output distributions of the mechanism given two different inputs.  We formulate an optimization problem to minimize this divergence subject to a cost constraint on the noise. We first prove that additive mechanisms are optimal. Since the optimization problem is infinite dimensional, it cannot be solved directly; nevertheless, we quantize the problem to derive near-optimal additive mechanisms that we call ``cactus mechanisms'' due to their shape. We show that our quantization approach can be arbitrarily close to an optimal mechanism. Surprisingly, for quadratic cost, the Gaussian mechanism is strictly sub-optimal compared to this cactus mechanism.  Finally, we provide numerical results which indicate that cactus mechanism outperforms the Gaussian mechanism for a finite number of compositions.
\end{abstract}
\blfootnote{This material is based upon work supported by the National Science Foundation under Grant Nos. CAREER-1845852, CIF-1900750, CIF-1815361, CIF-1901243, CIF-1908725, CIF-2007688, and SaTC 2031799. The authors also thank Oracle Research for a gift that supported this work.}
\noindent
This paper is Part I in a pair of papers, where Part II is~\cite{Fisher}.

\section{Introduction}

Likelihood ratios are at the heart of most privacy metrics. Consider the problem of quantifying the privacy loss suffered by  a sensitive variable $X$ given an observation of a disclosed variable $Y$. For example, $X$ may represent a dataset and $Y$ a randomized function computed over $X$. Privacy can be measured in terms of properties of the \emph{privacy loss random variable}, defined as
\begin{equation}
L_{x,x'}\coloneqq \log\frac{dP_{Y|X=x}}{dP_{Y|X=x'}}(Y),
\end{equation}
where $Y\sim P_{Y|X=x}$ and $x,x'\in \mathcal{X} \coloneqq \mathsf{supp}(X).$  The channel $P_{Y|X}$ is often referred to as a \emph{privacy mechanism}. 

Today, the most popular privacy definition (including, in practice~\cite{erlingsson2014rappor,Apple_Privacy,Facebook020GuidelinesFI}) is \emph{differential privacy} (DP), which quantifies privacy in terms of $L_{x,x'}$ when $x,x'$ are close or ``neighboring.'' Thus, given a metric $d:\mathcal{X}\times \mathcal{X}\to \Reals$, $P_{Y|X}$ is said to be $(\eps,\delta)$-differentially private ($(\varepsilon,\delta)$-DP)~\cite{Dwork_Calibration} if 
\begin{equation}\label{eq:DP0}
    \sup_{d(x,x')\leq s}\  \sup_{A\subset \mathcal Y}\left[P_{Y|X=x}(A) - e^\eps P_{Y|X=x'}(A)\right]\leq \delta,
\end{equation}
where $s$ determines when inputs $x$ and $x'$ are neighboring, and $\calY:=\supp(Y)$. Intuitively, if a mechanism is  $(\eps, \delta)$-differentially private for sufficiently small $\eps$ and $\delta$, then an adversary observing $Y$ cannot accurately distinguish between small changes in $X$.  

Most privacy mechanisms are applied several times on sensitive data. Quantifying privacy guarantees under multiple \emph{compositions} of a mechanism is a challenging problem. In the simple case where  the same mechanism $P_{Y|X}$ is independently applied $n$ times on data $X$ generating output $Y^n$, i.e., $P_{Y^n|X}=\prod_{i=1}^n P_{Y_i|X}$,  the  privacy loss random variable is given by
\begin{equation}
L_{x,x'}^n\coloneqq \sum_{i=1}^n \log\frac{dP_{Y_i|X=x}}{dP_{Y_i|X=x'}}(Y_i),
\end{equation}
where $Y_i\sim P_{Y_i|X=x}.$ Differential privacy can be cast in terms of the privacy loss random variable. The reader can directly verify that $n$ independent applications of a mechanism $P_{Y|X}$ is $(\eps, \delta)$-DP if 
\begin{equation}
    \sup_{d(x,x')\leq s} \mathbb{E}\left[ \left(1-e^{-(L_{x,x'}^n-\eps)} \right)^{+} \right]\leq  \delta.
\end{equation}
From the law of large numbers, the distribution of $L_{x,x'}^n/n$ will concentrate around its mean, the KL-divergence, as
\begin{equation}\label{eq:PRVmean}
    \frac{1}{n}\mathbb{E}\left[L_{x,x'}^n\right]=D\left(P_{Y|X=x}\|P_{Y|X=x'}\right).
\end{equation}
Since the function $f(u)\coloneqq (1-e^{-nu+\eps})^+$ is non-decreasing, in the limit of large compositions, privacy mechanisms with lower values of  $D(P_{Y|X=x}\|P_{Y|X=x'})$ will enjoy stronger $(\eps,\delta)$-DP guarantees. Thus, regardless of the exact distribution of the privacy loss random variable, its mean~\eqref{eq:PRVmean} plays a central role in the privacy guarantees offered after many compositions. In applications such as privacy-ensuring machine learning, the number of compositions frequently exceeds~$n=10^3$.

We study the design of privacy mechanisms with favorable $(\eps,\delta)$-DP guarantees under a large number of compositions. Our approach departs from previous work in that we focus on the large-composition regime instead of optimizing \eqref{eq:DP0}. Since after many compositions, privacy will be mostly determined by the mean of the privacy loss random variable \eqref{eq:PRVmean}, we solve the optimization problem
\begin{equation}\label{fe}
\begin{array}{ll}
\displaystyle \inf_{P_{Y|X}\in \sR}
& \displaystyle \displaystyle \sup_{|x-x'|\le s} D(P_{Y|X=x} \| P_{Y|X=x'}) \\[2ex]
\text{subject to} 
& \displaystyle \sup_{x\in \mathbb{R}} ~ \mathbb{E}[c(Y-x)\mid X=x] \le C,
\end{array}
\end{equation}
where $c:\mathbb{R}\to [0,\infty)$ is a pre-specified cost function,  $s,C>0$ are constants, and  $\sR$ is the set of all Markov kernels on $\mathbb{R}$. Note that the cost function  is critical: without the constraint,~\eqref{fe} can be trivially solved by any mechanism that is independent of $X$.

Our main contributions are as follows:
\begin{enumerate}[leftmargin=12pt]
    \item We show (Thm.~\ref{gb}) that additive mechanisms---i.e.,  where $Y=X+Z$ for a noise variable $Z$ independent of $X$---suffice to minimize~\eqref{fe}.
    \item Even restricting to additive mechanisms, \eqref{fe} is an infinite-dimensional optimization problem, so it cannot be solved directly. Instead, we formulate an approximate problem that is finite dimensional and can be solved efficiently. We prove (Thm.~\ref{gk}) that this approximate problem can get arbitrary close to optimal.
    \item We solve the approximate problem to derive (near) optimal mechanisms for the quadratic cost function, i.e., $c(x)=x^2$. We dub the resulting mechanism the ``cactus mechanism'' due to the shape of the distribution (see Fig.~\ref{fig:cactus}). Surprisingly, the Gaussian distribution is strictly sub-optimal for \eqref{fe}, as the cactus mechanism achieves a smaller KL divergence for the same variance.
    \item We bound the $(\varepsilon,\delta)$-DP for the cactus mechanism in the context of sub-sampled stochastic gradient descent using the moments accountant method. Compared to the same analysis applied to a Gaussian mechanism, our approach does better for a reasonable number of compositions.
\end{enumerate}

\subsection{Related Work}

Identifying optimal mechanisms is a fundamental and challenging problem in the domain of differential privacy. There have been several works in the literature that have attempted to address this problem. For instance, within the class of additive noise mechanisms and under the single shot setting (i.e., no composition), Ghosh et al.~\cite{ghosh2012universally}
showed that the geometric mechanism is universally optimal for $(\eps,0)$-DP in a Bayesian framework, and Gupte and
Sundararajan~\cite{Gupte_universally} derived the optimal noise distribution in a minimax cost framework. For a rather general cost function, the optimal noise distribution was shown to have a staircase-shaped density function~\cite{geng2015optimal,geng2015staircase, OptimalDP2}.

Geng and Viswanath~\cite{Geng_OptimalApproximate} showed that for $(\eps, \delta)$-DP and integer-valued query functions, in the single-shot setting,  the discrete uniform noise distribution and the discrete Laplacian noise distribution are asymptotically optimal (for $L^1$ and $L^2$ costs) within a constant multiplicative gap in the high privacy regime (i.e., both $\eps$ and $\delta$ approach zero). Geng et al.~\cite{Geng_truncatedLaplace} studied the same setting except for real-valued query functions and identified truncated Laplace distribution is asymptotically optimal in various high privacy regimes. Finally, Geng et al.~\cite{Geng_Uniform}  showed that the optimal noise distribution for real-valued query and $(0, \delta)$-DP is uniform  with probability
mass at the origin. Our work differs from these works in that we focus on the optimal mechanisms under a large number of compositions, rather than the single shot setting.

When considering a composition of $n$ mechanisms, an important line of research has been to derive tighter composition results: relationships between the DP parameters of the composed mechanism and the parameters of each constituent mechanism. There are several composition results in the literature, such as~\cite{dwork2010boosting,Vadhan_Murtagh,kairouz_Composition, abadi2016deep, asoodeh2021three, Mohammadi_Bucket}. More recently, Dong et al.\ \cite{dong2019gaussian} have proposed a composition result for large $n$ and for a new variant of DP, called Gaussian-DP, that leverages the central limit theorem. These results can be sub-optimal (see, for example, \cite[Fig. 1]{gopi2019numerical}). Consequently, numerical composition results have gained increasing traction as they lead to easier, yet powerful, methods for accounting the privacy loss in composition~\cite{koskela2020computing, gopi2019numerical, koskela21a_FFT, zhu2021optimal}. In particular, 
Koskela et al.~\cite{koskela2020computing} obtained a numerical composition result based on a numerical approximation of an integral that gives the DP parameters of the composed mechanism.  The approximation is carried out by discretizing the integral and by evaluating discrete convolutions via the fast Fourier transform algorithm. The running time and memory needed for this approximation were subsequently improved~\cite{gopi2019numerical}. While our work shares the focus on the large composition regime, we are primarily interested in synthesizing optimal mechanisms rather than analyzing existing mechanisms.

\subsection{Notation}

The Lebesgue measure on $\BR$ is denoted by $\lambda$. We denote by $\sR$ the set of all Markov kernels\footnote{It is true that any conditional distribution from $\BR$ into $\BR$ has a version that is a Markov kernel~\cite[Chapter 4, Theorem 2.10]{Erhan_book}.} on $\BR$, i.e., conditional distributions $P_{Y|X}$ for $\BR$-valued $X$ and $Y$ such that $x\mapsto P_{Y|X=x}(B)$ is a Borel function for all Borel sets $B\subset \BR$. The set $\sB$ denotes all Borel probability measures on $\BR$. We fix a real-valued random variable $X$ throughout, and let $P_X\in \sB$ be its induced Borel probability measure. The KL-divergence is denoted by $D(P\|Q)$, and also by $D(p\|q)$ if $P,Q\ll \lambda$ with densities $p$ and $q$. The expectation is denoted by $\BE_P[f]:=\int_{\BR}f \, dP$, and also by $\BE_p[f]$ if $P\ll \lambda$ has probability density function (PDF) $p$. We let $T_a$ denote the shift operator, i.e., for a function $f$ of a real variable the function $T_af$ is defined as $(T_af)(x) := f(x-a)$, and for a measure $P$ the measure $T_aP$ is defined by $(T_aP)(B) := P(B-a)$.

\section{Optimality of Additive Continuous Channels} \label{Sec:characterizations}

We start by deriving characterizations of  solutions to the optimization problem~\eqref{fe}. The difficulty of this problem lies in the fact that we are optimizing over all conditional distributions. This not only makes the problem infinite-dimensional, but it also renders direct approaches ineffective. The main result of this section, shown in Theorem~\ref{gb}, is that it suffices to consider continuous additive channels. In other words, the optimization in~\eqref{fe} may be restricted to conditional distributions of the form $P_{Y|X=x} = T_xP$ for some Borel probability measure $P$ on $\BR$ that is absolutely continuous with respect to the Lebesgue measure. Equipped with this reduction, we build in the next section an explicit family of finitely-parametrized distributions that are also optimal in~\eqref{fe}.

\subsection{Assumptions and Definitions}

Throughout the paper, we require the cost function to satisfy the following properties.
\begin{assumptions} \label{assumption:c1}
The cost function $c:\mathbb{R}\to\mathbb{R}$ satisfies:
\begin{itemize}
\item\emph{Positivity:} $c(x)\ge 0$ for all $x\in\mathbb{R},$ and $c(0)=0.$
\item\emph{Symmetry:} $c(x)=c(-x)$ for all $x\in\mathbb{R}$.
\item\emph{Monotonicity:} $c(x)\le c(x')$ if $|x|\le |x'|$.
\item\emph{Continuity:} $c$ is continuous over $\BR.$
\item\emph{Tail regularity:} There exist $\alpha,\beta> 0$ such that $c(x) \sim \beta x^\alpha$ as $x\to \infty.$
\end{itemize}
\end{assumptions}
A natural choice of cost function is the quadratic cost $c(x)=x^2$, but we allow $c(x)$ to be any function that satisfies the above assumptions. For example, $c(x)=|x|^\alpha$ for any positive $\alpha$ is a natural family of cost functions.

Let $\sP\subset \sR$ be the set of conditional distributions $P_{Y|X}$ satisfying the cost constraint in~\eqref{fe}, i.e., set
\begin{equation}
\sP := \left\{ P_{Y|X} \in \sR  ~; ~  \sup_{x\in \mathbb{R}} ~ \mathbb{E}[c(Y-x)\mid X=x] \le C \right\}.
\end{equation}
The infimal value in~\eqref{fe} is then
\begin{equation}\label{continuous_optimization}
    \text{KL}^\star := \inf_{P_{Y|X}\in\mathscr{P}}\ \sup_{x,x'\in\mathbb{R}:|x-x'|\le s}
    \ D(P_{Y|X=x}\|P_{Y|X=x'}).
\end{equation}
We are interested in computing $\text{KL}^\star,$ as well as mechanisms $P_{Y|X}$ that approach this optimal value. Note that, for clarity of presentation, we suppress the dependence on $(s,c,C)$ in the notations $\sP$ and $\text{KL}^\star.$

In the main problem~\eqref{fe}, we allow $P_{Y|X}$ to be any mechanism that produces $Y$ given $X$. A more restrictive but natural and easy-to-implement class of mechanisms is the \emph{additive} mechanism class. An additive mechanism is given by $P_{Y|X=x}(B)=T_xP(B)$ where $P$ is a Borel probability measure on $\mathbb{R}$. In other words, an additive mechanism $P_{Y|X}$ has $Y$ of the form $Y = X+Z$ for some noise random variable $Z\sim P \in \sB$ that is independent of the input $X.$ Let $\sP_{\text{add}}\subset\sB$ be the set of additive mechanisms satisfying the cost constraint in~\eqref{fe},
\begin{equation}
    \sP_{\text{add}} := \left\{ P \in \sB ~ ; ~ \BE_P[c] \le C \right\}.
\end{equation}
Since the KL-divergence is shift-invariant, restricting the optimization~\eqref{fe} to additive mechanisms amounts to considering the simplified optimization problem
\begin{equation} \label{gi}
    \text{KL}_{\text{add}}^\star
    :=\inf_{P\in\scP_{\text{add}}}\ \sup_{a\in\mathbb{R}:|a|\le s}\ D(P\|T_a P).
\end{equation}
Of course, it is immediate that $\text{KL}^\star \le \text{KL}_{\text{add}}^\star$. In fact, we will show below that these quantities are the same, meaning that there is no loss in restricting to additive mechanisms.

\subsection{Optimality of Continuous Additive Mechanisms}

The optimization problem in~\eqref{fe} is a convex problem, but the fact that the feasible set $\mathscr{P}$ is of infinite dimension means it cannot be solved directly, nor do the tractable properties one expects of a convex optimization problem necessarily follow. For example, in any finite dimensional convex optimization problem, a symmetry in the problem leads to the same symmetry in the solution. In this problem, one can see that shifting the mechanism---i.e., given $P_{Y|X}$, construct $Q_{Y|X=x}(B)=P_{Y|X=x+z}(B+z)$ for some $z$---does not change the cost constraint nor the objective value in \eqref{fe}. Thus, one might be inclined to conclude that the optimal mechanism is invariant to a shift (i.e., is an additive mechanism). Unfortunately, the infinite-dimensional nature of the problem means that this conclusion is not immediate. We resolve this issue in the following theorem which states that additive mechanisms are in fact optimal in \eqref{fe}.

\begin{theorem} \label{gb}
We have that 
\begin{equation}
    \textup{KL}^\star = \textup{KL}_{\textup{add}}^\star,
\end{equation}
and there exists a $P^\star \in \sP_{\textup{add}}$ achieving this value. Further, any such $P^\star$ is necessarily absolutely continuous.
\end{theorem}
\begin{proof}[Proof sketch]
The proof is given in Appendix~\ref{gg}. 
We give here only a high level description of the approach. Let $P_{Y|X}^{(k)}$ be a sequence achieving $\text{KL}^\star$. We make these mechanisms increasingly closer to being additive, while sacrificing neither feasibility nor utility, by considering the convex combinations
\begin{equation}
    \overline{P}_{Y|X=x}^{(k)}(B) := \BE\left[ P_{Y|X=x+Z_k}^{(k)}(B+Z_k) \right]
\end{equation}
where $Z_k \sim \mathrm{Unif}([-k,k])$. Specifically, one can invoke Prokhorov's theorem on the $\overline{P}_{Y|X}^{(k)}$, thereby extracting a probability measure $P^\star$ such that $\overline{P}_{Y|X=x}^{(k)} \to T_x P^\star$ weakly for each fixed $x$. Finally, we show that the mechanism $P^\star$ is optimal by invoking joint convexity and lower-semicontinuity of the KL-divergence.
\end{proof}
\begin{remark}
The proof of $P^\star \ll \lambda$ only relies on the property that $P^\star \ll T_a P^\star$ for every $|a|\le s$, which holds in view of $\text{KL}^\star < \infty$. Therefore, any \emph{feasible} additive mechanism must be absolutely continuous with respect to the Lebesgue measure, i.e., if $\mu \in \sB$ satisfies $\sup_{|a|\le s} D(\mu \| T_a\mu)<\infty$ then we necessarily have $\mu \ll \lambda$.
\end{remark}

\section{Numerical Approximation: The Cactus Distribution} \label{ej}

The optimization problem over additive mechanisms in~\eqref{gi} is  infinite-dimensional, so it cannot be solved numerically as-is, and it appears to have no closed-form solution for non-trivial cost functions. The lack of closed-form solution is true even for the simple case of $c(x)=x^2$: to our surprise, as will be illustrated later, the Gaussian mechanism is not optimal!\footnote{Of course, simply because Gaussian is not optimal does not imply that there is no closed-form solution. It is possible to write a set of KKT conditions for~\eqref{gi}, which we have omitted from this paper in the interest of space. This set of KKT conditions cannot be solved in closed-form.} In our companion paper~\cite{Fisher}, we explore the regime where $s\to 0^+$; in this limit, we show that the optimal distribution can be determined exactly, and in fact for quadratic cost the limiting optimal distribution \emph{is} Gaussian---although for other costs the optimal distribution is much more surprising. 

In the regime of fixed positive $s$, to find practically achievable near-optimal mechanisms, we resort to numerical approximation of~\eqref{gi}. In this section, we fix $s=1$. We can do this without loss of generality simply by scaling: that is, the optimization problem in~\eqref{gi} with sensitivity $s$ and cost function $c(x)$ is equivalent to the same problem with sensitivity $1$ and cost function $c(sx)$.

To approximate~\eqref{gi} by a numerically tractable problem, we (i) quantize the distribution, and (ii) only explicitly parameterize the distribution in a certain interval. Specifically, we construct a mapping from finite-length vectors to continuous distributions as follows. 

\begin{defn} \label{def:cactus}
Fix two positive integers $n$ and $N$, and a constant $r\in(0,1)$. Consider the partition of $\BR$ by intervals $\{ \calJ_{n,i} \}_{i\in \BZ}$ defined by: $\calJ_{n,0} := [-1/(2n),1/(2n)]$ and
\begin{equation}
    \calJ_{n,i}:= \left\lbrace \begin{array}{cl} 
    \left( \frac{i-1/2}{n}, \frac{i+1/2}{n} \right], & \text{if } i>0, \vspace{1mm} \\ 
    \left[ \frac{i-1/2}{n}, \frac{i+1/2}{n} \right), & \text{if }  i<0.
    \end{array} \right.
\end{equation}
We associate to each vector $\bp=(p_{0},p_{1},\ldots,p_N)  \in [0,1]^{N+1}$ a piecewise constant function that is defined by 
\begin{equation}\label{continuous_mechanism}
f_{n,r,\bp}(x)=\begin{cases} n p_{|i|}, & \text{if } x\in \calJ_{n,i}, \text{with } |i|< N,\\
np_{N} r^{|i|-N}, & \text{if } x\in \calJ_{n,i}, \text{with } |i|\ge N.
\end{cases}
\end{equation}
We also associate with $f_{n,r,\bp}$ the Borel measure $P_{n,r,\bp}$, where
\begin{equation}
    P_{n,r,\bp}(B) := \int_B f_{n,r,\bp}(x) \, dx.
\end{equation}
\end{defn}
\begin{remark}
Note that
\begin{equation} \label{bc}
    \int_{\BR} f_{n,r,\bp}(x) \, dx = p_0+\sum_{i=1}^{N-1}2p_i+\frac{2p_N}{1-r} =:S_{r,\bp}.
\end{equation}
If $S_{r,\bp}=1,$ then $P_{n,r,\bp}$ is a probability measure with density $f_{n,r,\bp}$. This distribution is symmetric around the origin, i.e., $f_{n,r,\bp}(x)=f_{n,r,\bp}(-x)$. Further, its tails decay almost geometrically: for $(N+1/2)/n < x_1 < x_2$ one has $f_{n,r,\bp}(x_2)=r^{nk}\cdot f_{n,r,\bp}(x_1)$ where $k= \left( \lceil nx_2-1/2 \rceil - \lceil nx_1 - 1/2 \rceil \right)/n \approx x_2-x_1$.
\end{remark}

The main results of this section are: we show that the distribution family introduced in Definition~\ref{def:cactus} is optimal for~\eqref{fe}, and we show that the optimal distribution \emph{within} this family (which we will call the \emph{cactus distribution}) is obtainable via a tractable finite-dimensional convex optimization problem.

We use the following notation. Consider the restriction of~\eqref{gi} to the mechanisms constructible by Definition~\ref{def:cactus}. For a fixed triplet $(n,N,r)\in \BN^2 \times (0,1)$, consider the set of mechanisms $\sC_{n,N,r}\subset \sB$,
\begin{equation}
    \sC_{n,N,r} := \left\{ P_{n,r,\bp} ~ ; ~  \bp \in [0,1]^{N+1}, S_{r,\bp} = 1 \right\}.
\end{equation}
(Recall the definition of $S_{r,\bp}$ from \eqref{bc}.) Denote the optimal value achievable by the class $\sC_{n,N,r}$ by
\begin{equation} \label{gj}
    \text{KL}^\star_{n,N,r}(C) := \inf_{\substack{P \in \sC_{n,N,r} \\
    \BE_P[c] \le C}} ~ \sup_{|a|\le 1} D(P\|T_a P).
\end{equation}

We show next that we may restrict the shift $a$ in~\eqref{gj} to take values over the finite set $\{1/n,2/n,\cdots,1\}$ (rather than varying over the whole interval $[-1,1]$), thereby rendering~\eqref{gj} a finite-dimensional optimization problem amenable to standard numerical convex-programming methods.

For each $i\in \BZ,$ we denote the constants
\begin{equation} \label{mh}
c_{n,i} := \int_{\calJ_{n,i}} nc(x) \, dx.
\end{equation}

\begin{theorem} \label{eh}
Fix $r\in(0,1)$, and positive integers $n<N$. The minimization~\eqref{gj} can be recast as the following convex program over the variable $\bp=(p_0,\cdots,p_N)\in \BR^{N+1}$
\begin{align}
&\underset{\bp}{\textup{minimize}} && \displaystyle 
\max_{k\in\{1,\ldots,n\}} \
\frac{1}{2}\sum_{i=-N+1}^{N-k-1} (p_{|i|}-p_{|i+k|})\log\frac{p_{|i|}}{p_{|i+k|}}
 +\sum_{i=N-k}^{N-1} (p_i-p_Nr^{i+k-N})\log \frac{p_i}{p_Nr^{i+k-N}} \nonumber \\
&&& \displaystyle \qquad+ p_N\frac{1- r^{k}}{1-r} k\log r^{-1} \nonumber \\
&\textup{subject to} &&  \displaystyle p_0c_{n,0}+\sum_{i=1}^{N-1} 2p_ic_{n,i}+2p_N \sum_{i=N}^\infty c_{n,i}r^{i-N}\le C, \nonumber \\
&&& \displaystyle p_0+\sum_{i=1}^{N-1}2p_i+\frac{2p_N}{1-r}=1,  \nonumber \\
&&&  \displaystyle p_i\ge 0\textup{ for all }i\in\{0,\ldots,N\}. \label{finite_d_opt}
\end{align}
\end{theorem}
\begin{proof}
See Appendix~\ref{ei}.
\end{proof}

Figure~\ref{fig:cactus} shows an example of the distribution that results from the finite-dimensional optimization problem in~\eqref{finite_d_opt} with a quadratic cost. The shape of this distribution\footnote{In addition to the state of Arizona being home of several of the authors.} 
has inspired the name the ``cactus distribution.'' The following result shows that cactus mechanisms derived from the optimization problem~\eqref{finite_d_opt} are in fact globally optimal for the main optimization problem~\eqref{fe}.

\begin{theorem} \label{gk}
Denote the optimal value a cactus distribution can achieve by
\begin{equation} \label{sa}
    \KL_{\textup{Cactus}}^\star := \lim_{\varepsilon\to 0^+} \inf_{(n,N,r)\in \BN^2\times (0,1)} \KL_{n,N,r}^\star(C+\varepsilon).
\end{equation}
We have that $\KL^\star = \KL_{\textup{Cactus}}^\star$.
\end{theorem}

\newpage

\begin{proof}
See Appendix~\ref{ft}.
\end{proof}

\begin{remark}
The proof of Theorem~\ref{gk} gives some guidelines for choosing the parameters $(n,N,r)$. For example, optimal cactus distributions can be obtained by restricting the ratio $N/n$ (chosen sufficiently large), and choosing $r=1-\Theta_\alpha(N^{-1})$. 
\end{remark}

\begin{figure}
    \centering
    \psfrag{x}{z}
    \psfrag{p(x)}{p(z)}    
    \includegraphics[width=0.65\columnwidth]{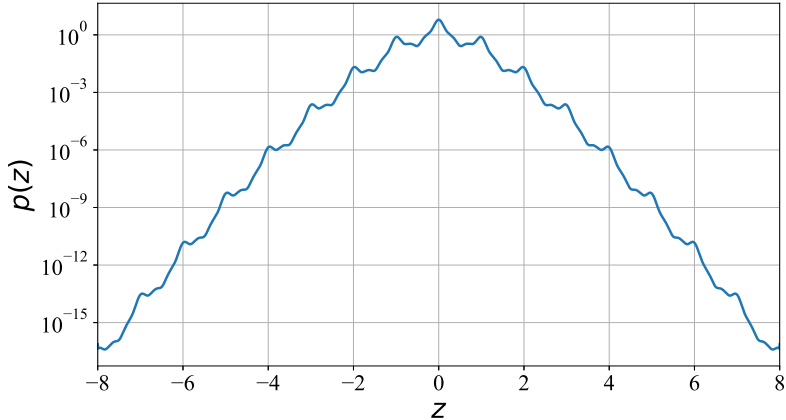}
    \caption{The optimal distribution $p(z)$, found by solving~\eqref{finite_d_opt} (and dubbed the \emph{cactus distribution}), plotted on a semi-log scale. The cost function is $c(z)=z^2$, and the parameters are: $s=1$, $C=0.25$, $n=200$, $N=1600$, and $r=0.9$.}
    \vspace{2.5mm}
    \label{fig:cactus}
\end{figure}

\section{Numerical Results}\label{sec:experiments_cactus}

We solve the optimization problem \eqref{finite_d_opt} using an interior-point method. An example of the cactus distribution for quadratic cost is shown in Figure~\ref{fig:cactus}. Figure~\ref{fig:gaussian_comparison} compares the maximal KL-divergence achieved by the cactus to that of Gaussian distributions for fixed sensitivity $s=1$ and various $\sigma$. As noted above, varying $\sigma$ with fixed $s$ is equivalent to varying $s$ with fixed $\sigma$. The KL-divergence for cactus is computed numerically, and for Gaussian mechanisms the KL-divergence is exactly $\frac{1}{2\sigma^2}$.
The cactus distribution outperforms the Gaussian distribution in terms of KL-divergence for all values of $\sigma$, although the difference decreases as $\sigma$ grows such that for larger values of $\sigma$ it is difficult to discern any gap between the curves in Figure~\ref{fig:gaussian_comparison}. (Our companion paper \cite{Fisher} gives a theoretical explanation for why Gaussian is so close to optimal as $s/\sigma$ decreases.) To illustrate that this improvement in KL-divergence leads to an improvement in $(\varepsilon,\delta)$-DP, we compute the achieved privacy via moments accountant~\cite{abadi2016deep} for each mechanism. Figure~\ref{fig:L2_comparison} shows the resulting $\varepsilon$ value as a function of the number of compositions, for fixed $\delta=10^{-3}$. Indeed, the cactus mechanism does better than Gaussian.

To give a reasonable comparison in the context of machine learning, we modified the tutorial code in TensorFlow-Privacy~\cite{tfprivacy}, which implements the DP-stochastic gradient descent (SGD) algorithm with a Gaussian  mechanism on a convolutional neural network (CNN) model. We use the training results from the original tutorial as a benchmark, then replace the Gaussian mechanism with our cactus mechanism, and train the model using the renewed setting. We select a noise level $\sigma=\sqrt{0.1}$. We test the original and modified model on a popular image dataset, MNIST, which is of size $60000$. We choose a batch-size $250$, such that each epoch consists of $240$ iterations (i.e., compositions) and the sub-sampling rate\footnote{The cactus mechanism is not optimized for subsampling. Nevertheless, we observe numerical performance of the cactus mechanism in the subsampling setting outperforming that of the Gaussian mechanism.}  is $q=250/60000\approx 0.00417$. Figure~\ref{fig:MNIST_comparisonA} shows the achieved $(\varepsilon,\delta)$-DP as computed by the moments account in this setting. Fixing $\delta=10^{-5}$, Figure~\ref{fig:MNIST_comparisonB} shows the tradeoff between privacy $\varepsilon$ and accuracy of the resulting CNN as the number of training iterations increases. One can see that for a fixed privacy budget 
(i.e., fixed $\eps$ and $\delta$), the cactus mechanism allows more training iterations and, thus, better accuracy.

\begin{figure}
    \centering
    \includegraphics[width=0.6\textwidth]{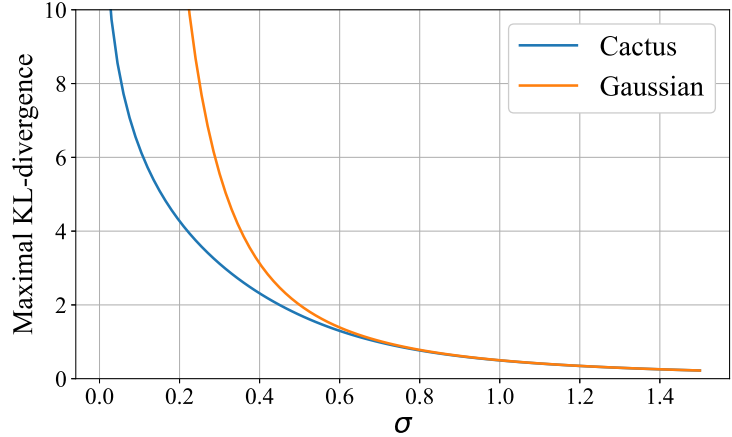}
    \caption{Achieved maximal KL-divergence $\sup_{|a|\le s} D(p\|T_ap)$ versus $\sigma$, the (quadratic) cost constraint is of the form $\mathbb{E}[Z^2]\le \sigma^2 = C$ with fixed sensitivity $s=1$.}
    \label{fig:gaussian_comparison}
\end{figure}

\begin{figure}
    \centering  
    \includegraphics[width=0.6\textwidth]{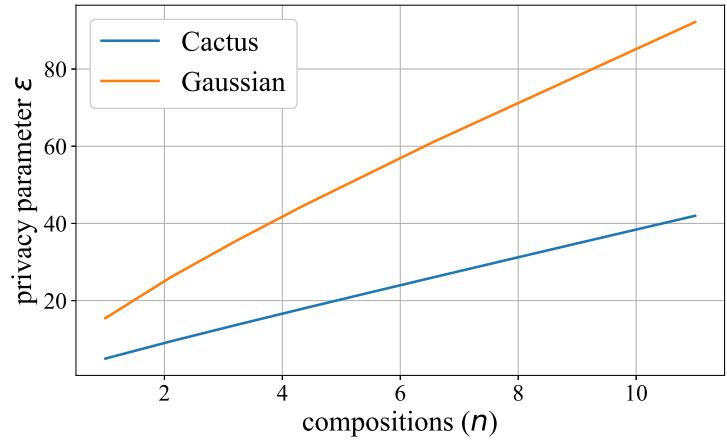}
    \caption{Privacy parameter $\eps$ versus the number of compositions, computed via the moments accountant, where $\delta=10^{-3}$, and quadratic cost $C=0.1$ with fixed sensitivity $s=1$.}
    \label{fig:L2_comparison}
\end{figure}

\begin{figure}
    \centering  
    \includegraphics[width=0.6\textwidth]{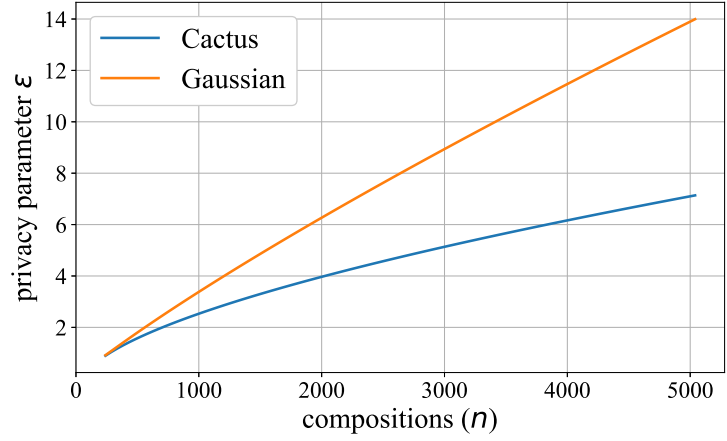}
    \caption{Privacy parameter $\eps$ versus the number of compositions, computed via the moments accountant, where $\delta=10^{-5}$, subsampling rate $q\approx 0.00417$, and quadratic cost $C=0.1$ with fixed sensitivity $s=1$.}
    \label{fig:MNIST_comparisonA}
\end{figure}

\begin{figure}
    \centering  
    \includegraphics[width=0.6\textwidth]{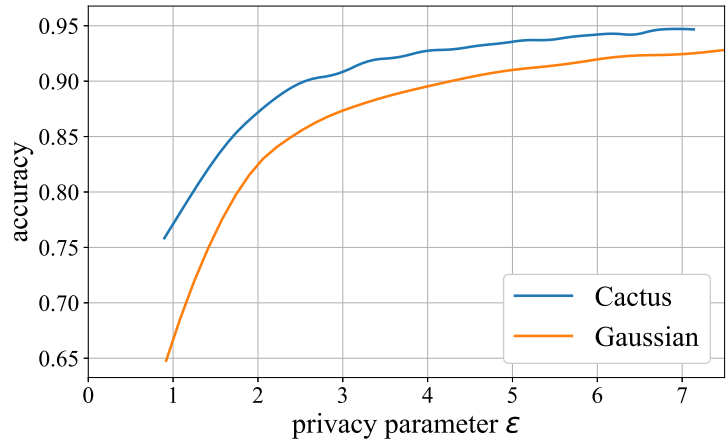}
    \caption{Model accuracy versus privacy parameter $\eps$. The settings are the same as in Figure~\ref{fig:MNIST_comparisonA} and experiment details are given in Section~\ref{sec:experiments_cactus}.}
    \label{fig:MNIST_comparisonB}
\end{figure}

\clearpage

\begin{appendices}

\section{Proof of Theorem~\ref{gb}: Optimality of Additive Continuous Channels} \label{gg}

Let $F:\sR \to [0,\infty]$ denote the objective function in~\eqref{fe}, i.e.,
\begin{equation}
    F(P_{Y|X}) :=  \sup_{|u-v|\le s} ~ D(P_{Y|X=u} \| P_{Y|X=v}).
\end{equation}
Thus, 
\begin{equation}
    \text{KL}^\star = \inf_{P_{Y|X}\in \sP} ~ F(P_{Y|X}).
\end{equation}
Fix a sequence of conditional distributions 
\begin{equation}
    \left\{P_{Y|X}^{(k)}\right\}_{k\in \BN}\subset \sP
\end{equation}
satisfying 
\begin{equation} \label{pn}
    \text{KL}^\star = \lim_{k\to \infty} F\left( P_{Y|X}^{(k)}\right).
\end{equation}
Recall that by assumption, the version of each conditional distribution $P_{Y|X}^{(k)}$ we choose is regular, i.e., $x\mapsto P_{Y|X=x}^{(k)}(B)$ is a Borel function for each Borel set $B\subset \BR$. Note that $\text{KL}^\star < \infty$ since, e.g., the Gaussian mechanism is feasible. Throwing away the first few elements in the sequence, we assume that $F\left(P_{Y|X}^{(k)}\right)<\infty$ for each $k\in \BN$. 

We break the proof down into several steps:

\begin{enumerate}
    \item Introduce Markov kernels $\overline{P}_{Y|X}^{(k)}$ as ``continuous'' convex combinations of the $P_{Y|X}^{(k)}$.
    
    \item The $\overline{P}_{Y|X}^{(k)}$ also satisfy the cost constraint.
    
    \item The $\overline{P}_{Y|X}^{(k)}$ asymptotically achieve $\text{KL}^\star$.
    
    \item The $\overline{P}_{Y|X=x}^{(k)}$ are asymptotically shifted versions $T_xP^\star$ of a fixed $P^\star \in \sB$.
    
    \item $P^\star$ achieves $\text{KL}^\star$. \label{se}
\end{enumerate}
\vspace{1mm}
\noindent $\bullet$ \underline{\emph{Step 1}}: Averaging the $P_{Y|X}^{(k)}$.

For $k\in \BN$, we will define the Markov kernel $\overline{P}_{Y|X}^{(k)}\in \sR$ by
\begin{equation} \label{sf}
    \overline{P}_{Y|X=x}^{(k)}(B) := \frac{1}{2k} \int_{-k}^k P^{(k)}_{Y|X=x+z}(B+z) \, dz.
\end{equation}
Of course, we need to check that~\eqref{sf} indeed yields a Markov kernel $\overline{P}_{Y|X}^{(k)}$. In view of Fubini's theorem, it suffices to check that the map $(x,z)\mapsto P^{(k)}_{Y|X=x+z}(B+z)$ is jointly Borel (for every fixed Borel set $B\subset \BR$). This joint measurability is not self-evident, so we check next that it indeed holds.

Let the transition probability kernel $L^{(k)}:\BR^2\times \calB(\BR)\to [0,1]$ be defined by
\begin{equation}
    L^{(k)}((x,z),A):=P_{Y|X=x+z}^{(k)}(A).
\end{equation}
Let $N^{(k)}: \BR^2 \times \calB(\BR) \to [0,1]$ denote the map
\begin{equation} \label{ph}
    N^{(k)}((x,z),B) := P^{(k)}_{Y|X=x+z}(B+z).
\end{equation}
For each $(x,z)\in \BR^2$ and Borel set $B\subset \BR$, we may write $N^{(k)}((x,z),B)$ as the integral of a nonnegative Borel function against $L((x,z), dy)$, namely,
\begin{equation} \label{pg}
    N^{(k)}((x,z),B) = \int_\BR 1_B(y-z) \, L((x,z),dy).
\end{equation}
Hence (see, e.g.,~\cite[Chapter 1, Proposition 6.9]{Erhan_book}) $(x,z) \mapsto N^{(k)}((x,z),B)$ is a Borel function. Hence, $\overline{P}_{Y|X}^{(k)}$ as given by~\eqref{sf} is indeed a well-defined Markov kernel on $\BR$.

For the next steps, we will use the following notation
\begin{align}
    R^{(k,x)}_{Y|X=z}(B) &:= P_{Y|X=x+z}^{(k)}(B+z), \label{sg} \\
    P^{(k,x)}(B) &:= \overline{P}_{Y|X=x}^{(k)}(B), \label{si} \\
    U^{(k)}(B) &:=\frac{1}{2k} \cdot \lambda(B\cap[-k,k]).
\end{align}
Note that $R^{(k,x)}_{Y|X}\in \sR$ and $P^{(k,x)}\in \sB$ for each fixed $(k,x)\in \BN\times \BR$, and~\eqref{sf} may be rewritten as
\begin{equation} \label{sh}
    P^{(k,x)} = R_{Y|X}^{(k,x)}\circ U^{(k)}.
\end{equation}

\noindent $\bullet$ \underline{\emph{Step 2}}: The $\overline{P}_{Y|X}^{(k)}$ satisfy the cost constraint.

Fix $k\in \BN$, and we will show next that $\overline{P}_{Y|X}^{(k)}\in \sP$, i.e., that $\overline{P}_{Y|X}^{(k)}$ satisfies the cost constraint. Recall that a Markov kernel $P_{Y|X}\in \sR$ belongs to $\sP$ if and only if it satisfies
\begin{equation} \label{sl}
    \sup_{x\in \BR} ~ \BE_{P_{Y|X=x}}\left[ T_x c \right] \le C.
\end{equation}
By the assumption that $P_{Y|X}^{(k)}\in \sP$, we have that
\begin{equation} \label{gp}
    \BE_{P_{Y|X=x}^{(k)}}\left[ T_xc \right] \le C
\end{equation}
for every $x\in \BR$. Shifting the variable of integration in~\eqref{gp} by a fixed constant $-z$, we obtain that
\begin{equation} \label{gq}
    \BE_{T_{-z}P_{Y|X=x}^{(k)}}\left[ T_{x-z}c \right] \le C
\end{equation}
for every $(x,z)\in \BR^2$. Replacing $x$ by $x+z$ in~\eqref{gq}, we conclude that (see~\eqref{sg})
\begin{equation} \label{pe}
    \BE_{R^{(k,x)}_{Y|X=z}}\left[ T_xc \right] \le C
\end{equation}
for every $(x,z)\in \BR^2$. We proceed via the following standard approximation by simple functions argument. 

Fix $x\in \BR$, and let $\sum_j a_j 1_{B_j}(y)$ be a nonnegative simple function upper bounded by $(T_xc)(y)$. Integrating against $R^{(k,x)}_{Y|X=z}(dy)$ we deduce from~\eqref{pe} that
\begin{equation} \label{sj}
    \sum_j a_j R^{(k,x)}_{Y|X=z}(B_j) \le C
\end{equation}
for every $z\in \BR$. Integrating~\eqref{sj} against $U^{(k)}(dz)$, and noting that $P^{(k,x)}=R^{(k,x)}_{Y|X}\circ U^{(k)}$ (see~\eqref{sh}), we deduce that 
\begin{equation} \label{pm}
    \sum_j a_j P^{(k,x)}(B_j) \le C.
\end{equation}
Now, as~\eqref{pm} holds for all nonnegative simple functions below $T_xc$, taking an increasing sequence of nonnegative simple function converging pointwise to $T_xc$ we conclude that
\begin{equation}
    \BE_{P^{(k,x)}}[T_xc] \le C.
\end{equation}
In other words (see~\eqref{si}), 
\begin{equation} \label{sk}
    \BE_{\overline{P}_{Y|X=x}^{(k)}}[T_xc] \le C.
\end{equation}
As~\eqref{sk} holds for all $x\in \BR$, we have shown that $\overline{P}_{Y|X}^{(k)}\in \sP$.

\noindent $\bullet$ \underline{\emph{Step 3}}: The $\overline{P}_{Y|X}^{(k)}$ are asymptotically optimal.

Next, we use monotonicity of the KL-divergence under conditioning (see Lemma~\ref{qb}) to show the limit
\begin{equation} \label{gr}
    \text{KL}^\star = \lim_{k\to \infty} F\left( \overline{P}_{Y|X}^{(k)} \right).
\end{equation}
Shift-invariance of the KL-divergence implies that, for each $x,x',z\in \BR$,
\begin{equation}
    D\left( R_{Y|X=z}^{(k,x)} \| R_{Y|X=z}^{(k,x')} \right) = D\left( P_{Y|X=x+z}^{(k)} \| P_{Y|X=x'+z}^{(k)} \right).
\end{equation}
Thus, as $(x+z)-(x'+z) = x-x'$, we conclude that
\begin{equation}
    \sup_{\substack{|x-x'|\le s \\
    z\in \BR}} D\left( R_{Y|X=z}^{(k,x)} \| R_{Y|X=z}^{(k,x')} \right) = F\left( P_{Y|X}^{(k)} \right)
\end{equation}
By assumption of optimality of the $P_{Y|X}^{(k)}$ (see~\eqref{pn}), there exists a $k_0$  such that for all $k\ge k_0$,
\begin{equation} \label{oy}
    \sup_{\substack{|x-x'|\le s \\
    z\in \BR}} D\left( R_{Y|X=z}^{(k,x)} \| R_{Y|X=z}^{(k,x')} \right) \le \text{KL}^\star + \delta.
\end{equation}
By definition of KL-divergence, we infer $R_{Y|X=z}^{(k,x)} \ll R_{Y|X=z}^{(k,x')}$ for all $z\in \BR$ and $|x-x'|\le s$. Also,~\eqref{oy} shows in particular that
\begin{equation}
    \sup_{|x-x'|\le s} \BE_{\xi \sim U^{(k)}} \left[ D\left( R_{Y|X=\xi}^{(k,x)} \| R_{Y|X=\xi}^{(k,x')} \right) \right] \le \text{KL}^\star + \delta.
\end{equation}
Using~\eqref{sh}, Lemma~\ref{qb} yields that
\begin{equation}
    \sup_{|x-x'|\le s} D\left( P^{(k,x)} \| P^{(k,x')} \right) \le \text{KL}^\star + \delta.
\end{equation}
Taking $\delta\to 0^+$, we see that~\eqref{gr} holds.

\vspace{1mm}

\noindent $\bullet$ \underline{\emph{Step 4}}: $P^{(k,x)}$ is asymptotically $T_xP^\star$ for a fixed $P^\star$.

Next, we show that there is a measure $P^\star\in \sB$ such that, for every $x\in \BR$, we have the weak convergence
\begin{equation}
    P^{(k,x)} \to T_x P^\star
\end{equation}
as $k\to \infty$. 

First, for each fixed $x\in \BR$, we establish the total-variation distance convergence
\begin{equation} \label{gw}
    \lim_{k\to \infty} \left\| P^{(k,x)} - T_xP^{(k,0)} \right\|_{\text{TV}} = 0.
\end{equation}
We may write
\begin{equation}
    \left(T_xP^{(k,0)}\right)(B) = \frac{1}{2k} \int_{-k-x}^{k-x} R_{Y|X=z}^{(k,x)}(B) \, dz.
\end{equation}
Therefore, for any Borel set $B\subset \BR$ we have that 
\begin{align}
    \left| P^{(k,x)}(B) - T_xP^{(k,0)}(B) \right|  \le \frac{1}{2k}\int_{[-k,k]\Delta [-k-x,k-x]} R_{Y|X=z}^{(k,x)}(B) \, dz \le \frac{|x|}{k}, \label{gv}
\end{align}
where $\Delta$ denotes the symmetric difference. As the bound~\eqref{gv} is uniform in $B$, we conclude that the total-variation limit in~\eqref{gw} holds.

The next ingredient we need is that the set $\{P^{(k,0)}\}_{k\in \BN} \subset \sB$ is tight, i.e., that for any $\varepsilon>0$ there exists an $n>0$ such that
\begin{equation} \label{gx}
    \sup_{k\in \BN} ~ P^{(k,0)}(\BR\setminus [-n,n]) \le \varepsilon.
\end{equation}
Fix $\varepsilon>0$. By the assumption that $c(x)\sim\beta |x|^\alpha$ where $\alpha,\beta>0$, we have $\lim_{|x|\to\infty} c(x)=\infty$. Thus there exists an integer $n$ such that $c(x)\ge C/\varepsilon$ whenever $|x|\ge n.$ Then, for each $(z,k)\in \BR\times \BN$,
\begin{align}
    R_{Y|X=z}^{(k,0)}\left( \BR \setminus [-n,n] \right) &= P_{Y|X=z}^{(k)}\left( \BR \setminus [-n+z,n+z] \right) \\
    &\le \int_{\BR \setminus [-n+z,n+z]} \frac{c(y-z)}{c(n)} \, dP_{Y|X=z}^{(k)}(y) \label{sm} \\
    &\le c(n)^{-1} \BE_{P_{Y|X=z}^{(k)}}[T_zc] \label{sn} \\
    &\le c(n)^{-1} \cdot C \label{so} \\
    &\le \varepsilon,
\end{align}
where~\eqref{sm} follows by monotonicity of $c$,~\eqref{sn} by nonnegativity of $c$, and~\eqref{so} since $P_{Y|X}^{(k)}\in \sP$. Hence, 
\begin{equation}
    \sup_{(z,k)\in \BR\times \BN} ~ R_{Y|X=z}^{(k,0)}\left(\BR\setminus [-n,n]\right) \le \varepsilon.
\end{equation}
Averaging over $z$, we deduce that~\eqref{gx} holds, i.e., that $\{P^{(k,0)}\}_{k\in \BN}$ is tight.

By tightness of $\{P^{(k,0)}\}_{k\in \BN}$, we conclude via Prokhorov's theorem~\cite[Chapter 3, Theorem 5.13]{Erhan_book} after passing to a subsequence that there is a $P^\star \in \sB$ such that $P^{(k,0)}\to P^\star$ weakly as $k\to \infty$, i.e., for every continuous and bounded function $f:\BR\to \BR$ we have
\begin{equation}
    \lim_{k\to \infty} \BE_{P^{(k,0)}}[f] = \BE_{P^\star}[f].
\end{equation}
This immediately implies that, for each $x\in \BR$, we also have 
\begin{equation} \label{gy}
    T_xP^{(k,0)} \to T_x P^\star
\end{equation}
weakly as $k\to \infty$. As convergence in total variation is stronger than weak convergence, we conclude from~\eqref{gw} and~\eqref{gy} that for every $x\in \BR$
\begin{equation} \label{gz}
    P^{(k,x)} \to T_x P^\star
\end{equation}
weakly as $k\to \infty$. \\

\noindent $\bullet$ \underline{\emph{Step 5}}: The additive mechanism $P^\star$ is optimal.

The final step is showing that $P^\star$ attains $\text{KL}^\star$ and satisfies the cost constraint. By joint lower-semicontinuity of the KL-divergence~\cite[Theorem 1]{posner}, we deduce from~\eqref{gz} that for each $x\in \BR$ 
\begin{equation} \label{ha}
    D(P^\star \| T_x P^\star) \le \liminf_{k\to \infty} D\left( P^{(k,0)} \| P^{(k,x)} \right).
\end{equation}
But we also have
\begin{equation}
    \sup_{|x|\le s} D\left( P^{(k,0)} \| P^{(k,x)} \right) \le F\left( \overline{P}_{Y|X}^{(k)} \right).
\end{equation}
Therefore, taking the supremum over $|x|\le s$ in~\eqref{ha}, we infer from~\eqref{gr} that
\begin{equation} \label{hc}
    \sup_{|x|\le s} D(P^\star \| T_x P^\star) \le \text{KL}^\star.
\end{equation}
Hence, it only remains to check that $P^\star \in \sP_{\text{add}}$ for us to conclude that equality holds in~\eqref{hc}.

For every $A>0$ and $x\in \BR$, the function $1_{[-A,A]}\cdot T_xc$ is continuous and bounded. Hence, the weak convergence $P^{(k,x)} \to T_xP^\star$ yields
\begin{equation} \label{hb}
    \BE_{T_xP^\star}\left[1_{[-A,A]}\cdot T_xc\right] = \lim_{k\to \infty} \BE_{P^{(k,x)}}\left[1_{[-A,A]}\cdot T_xc\right].
\end{equation}
As $\overline{P}_{Y|X}^{(k)} \in \sP$, nonnegativity of $c$ implies in view of~\eqref{hb} that
\begin{equation}
    \BE_{T_xP^\star}\left[1_{[-A,A]}\cdot T_xc\right] \le C.
\end{equation}
By the monotone convergence theorem, taking $A\to \infty$ yields
\begin{equation}
    \BE_{T_x P^\star} \left[ T_x c \right] \le C,
\end{equation}
In other words, $P^\star \in \sP_{\text{add}}$. Therefore, we must have
\begin{equation}
   \text{KL}^\star \le \text{KL}_{\text{add}}^\star \le \sup_{|x|\le s} ~ D(P^\star \| T_x P^\star) .
\end{equation}
Combining this inequality with~\eqref{hc}, we conclude that
\begin{equation} \label{qe}
   \text{KL}^\star = \text{KL}_{\text{add}}^\star = \sup_{|x|\le s} ~ D(P^\star \| T_x P^\star) .
\end{equation}
This completes the proof of the first statement of the theorem.

For the last statement of the theorem, we show that the relation $\mu \ll T_x \mu$ for every $|x|\le s$ (which holds for $P^\star$ by~\eqref{qe} and $\text{KL}^\star<\infty$) is enough to conclude that $\mu \ll \lambda$. Fix a Borel set $B\subset \BR$ such that $\lambda(B)=0$, and we will show that $\mu(B)=0$. Note that the function $x\mapsto (T_x \mu)(B)$ is Borel as it is given by the convolution $1_B \ast \eta$ where $\eta(A):=\mu(-A)$. Then, by Tonelli's theorem and translation-invariance of the Lebesgue measure,
\begin{align}
    \int_{\BR} (T_x\mu)(B) \, d\lambda(x) &= \int_{\BR^2} 1_{B-x}(b) \, d\mu(b) \, d\lambda(x) \\
    &= \int_{\BR^2} 1_{B-x}(b)  \, d\lambda(x) \, d\mu(b) \\
    &= \int_{\BR^2} 1_{B-b}(x) \, d\lambda(x) \, d\mu(b) \\
    &= \int_\BR (T_b\lambda)(B) \, d\mu(b) \\
    &= \int_\BR \lambda(B) \, d\mu(b) = 0.
\end{align}
Thus, $(T_x \mu)(B)=0$ for $\lambda$-almost every $x$. In particular, $(T_x \mu)(B)=0$ for at least one $x\in [-s,s]$. Thus, $\mu \ll T_x \mu$ implies $\mu(B)=0$, and the proof is complete.

\begin{remark}
The lemma stated below, showing that conditioning increases divergence, is a well-known fact. It is shown in the literature under various assumptions on the underlying distributions (see, e.g.,~\cite[Theorem~2.2 and Section~2.6]{Polyanskiy2019}). We use it in the proof of Theorem~\ref{gb} in the specific situation where one of the conditional distributions is absolutely continuous with respect to the other for each individual input. As in~\cite[Remark~2.4]{Polyanskiy2019}, Doob's version of the Radon-Nikodym theorem can be used to derive that conditioning increases divergence in our case. For completeness, we add a proof of this lemma here.
\end{remark}

\begin{lemma}[Conditioning increases divergence] \label{qb}
Let $P_{Y|X},P_{Y|X}'$ be Markov kernels on $\BR$ such that $P_{Y|X=x} \ll P_{Y|X=x}'$ for every $x\in \BR$. Then, denoting the marginalizations of $P_{X,Y}:=P_{Y|X}\otimes P_X, P_{X,Y}':=P_{Y|X}'\otimes P_X$ in the second coordinate by $P_Y,P_Y'$, we have that
\begin{equation} \label{qa}
    D\left( P_Y \| P_Y' \right) \le \BE_{\xi \sim P_X}\left[ D\left( P_{Y|X=\xi} \| P_{Y|X=\xi}' \right) \right].
\end{equation}
\end{lemma}
\begin{proof}
Since by assumption $P_{Y|X=x}\ll P_{Y|X=x}'$ for every $x\in \BR$, a generalization of the Radon-Nikodym theorem by Doob (see~\cite[Chapter 5, Theorem 4.44]{Erhan_book}) yields the existence of a version of the Radon-Nikodym derivatives $dP_{Y|X=x}/dP_{Y|X=x}'$ such that the function
\begin{equation}
    (x,y) \mapsto \frac{dP_{Y|X=x}}{dP_{Y|X=x}'}(y)
\end{equation}
is jointly measurable. We show that this function is a version of $dP_{X,Y}/dP_{X,Y}'$. First, note that $P_{X,Y}\ll P_{X,Y}'$ are equivalent. Indeed, for any Borel set $E\subset \BR$, denoting the sections by $E_x:=\{y\in \BR ~ ; ~ (x,y)\in E\}$, we have that $P_{X,Y}(E)=0$ if and only if $P_{Y|X=x}(E_x) = 0$ for $P_X$-a.e. $x$, and a similar statement holds for $P_{X,Y}'$. By assumption, $P_{Y|X=x}\ll P_{Y|X=x}'$ for each $x$, so we obtain $P_{X,Y} \ll P_{X,Y}'$. By joint measurability and nonnegativity, using the disintegration theorem (see, e.g.,~\cite[Chapter 1, Theorem 6.11]{Erhan_book}) we obtain that for any Borel $E\subset \BR^2$
\begin{align}
    \int_{E} \frac{dP_{Y|X=x}}{dP_{Y|X=x}'}(y) \, dP_{X,Y}'(x,y) & = \int_\BR \int_{E_x} \frac{dP_{Y|X=x}}{dP_{Y|X=x}'}(y) \, dP_{Y|X=x}'(y) \, dP_X(x) \\
    &  = \int_\BR \int_{E_x} dP_{Y|X=x}(y) \, dP_X(x) \\
    & = P_{X,Y}(E).
\end{align}
Thus, we have the equality
\begin{equation} \label{pw}
    \frac{dP_{X,Y}}{dP_{X,Y}'}(x,y) = \frac{dP_{Y|X=x}}{dP_{Y|X=x}'}(y)
\end{equation}
for $P_{X,Y}'$-a.e. $(x,y)$. 

Define $f:[0,\infty) \to [-1/e,\infty)$ by $f(0)=0$ and $f(t) = t\log t$ for $t>0$. By the disintegration theorem and~\eqref{pw}, we have the equality
\begin{align}
    D\left( P_{X,Y} \| P_{X,Y}' \right) & = \int_{\BR^2} f\left( \frac{dP_{X,Y}}{dP_{X,Y}'} \right) \, dP_{X,Y}' \\
    & = \int_\BR \int_\BR f\left( \frac{dP_{Y|X=x}}{dP_{Y|X=x}'}(y) \right) \, dP_{Y|X=x}' \, dP_X(x) \\
    & = \BE_{\xi \sim P_X}\left[ D\left( P_{Y|X=\xi} \| P_{Y|X=\xi}' \right) \right]. \label{pz}
\end{align}

On the other hand, disintegration with respect to $Y$ yields the following bound. Denote by $P_{X|Y},P_{X|Y}'$ the disintegrations of $P_{X,Y},P_{X,Y}'$ with respect to $P_Y,P_Y'$. In particular, $P_{X|Y}$ and $P_{X|Y}'$ are Markov kernels on $\BR$. By the disintegration theorem and Jensen's inequality,
\begin{align}
    D\left( P_{X,Y} \| P_{X,Y}' \right) &= \int_{\BR^2} f\left( \frac{dP_{X,Y}}{dP_{X,Y}'} \right) \, dP_{X,Y}' \\
    &= \int_\BR \int_\BR f\left( \frac{dP_{X,Y}}{dP_{X,Y}'} (x,y) \right) \, dP_{X|Y=x}'(x) \, dP_Y'(y) \\
    &\ge \int_\BR f\left( g(y) \right) \, dP_Y'(y) \label{px}
\end{align}
where 
\begin{equation}
    g(y) :=   \int_\BR \frac{dP_{X,Y}}{dP_{X,Y}'} (x,y) \, dP_{X|Y=x}'(x).
\end{equation}
For this application of Jensen's inequality, we use the fact, shown next, that $g$ is finite $P_Y'$-a.e. In fact, we show that $g$ is a version of $dP_Y/dP_Y'$. Note that $P_{X,Y}\ll P_{X,Y}'$ implies that $P_Y\ll P_Y'$. Now, for any Borel $B\subset \BR$, the disintegration theorem yields that
\begin{align}
    \int_B g \, dP_Y' &=  \int_B \int_\BR \frac{dP_{X,Y}}{dP_{X,Y}'} (x,y) \, dP_{X|Y=x}'(x) \, dP_Y'(y) \\
    &= \int_{\BR \times B} \frac{dP_{X,Y}}{dP_{X,Y}'} \, dP_{X,Y}' \\
    &= P_{X,Y}(\BR \times B) = P_Y(B).
\end{align}
Thus, we have that
\begin{equation}
    g(y) = \frac{dP_Y}{dP_Y'}(y).
\end{equation}
for $P_Y'$-a.e. $y$. Hence, we obtain from inequality~\eqref{px} that
\begin{equation} \label{py}
    D\left( P_{X,Y} \| P_{X,Y}' \right) \ge D\left( P_Y \| P_Y' \right).
\end{equation}
Combining inequality~\eqref{py} and equation~\eqref{pz} we obtain the desired inequality~\eqref{qa}.
\end{proof}


\section{Proof of Theorem~\ref{eh}: Finite-Dimensionality} \label{ei}

Note that the vector $\bp$ only includes $p_i$ for $0\le i\le N$. We will simplify our analysis by defining $p_i$ for all integers $i$. Specifically, for $i\in \BZ \setminus \{0,\cdots,N\}$, we denote
\begin{equation} \label{p_i_def}
    p_i := \left\{ \begin{array}{cl} 
    p_{|i|}, & \text{if } -N\le i \le -1, \vspace{2mm} \\
    p_Nr^{|i|-N}, & \text{if } |i|>N.
    \end{array} \right.
\end{equation}
Thus we may rewrite the formula for $f_{n,r,\bp}$ in~\eqref{continuous_mechanism} as
\begin{equation}\label{mech_quant}
f_{n,r,\bp}(x)=np_i \quad \text{ if } x\in \calJ_{n,i}.
\end{equation}
We show first that
\begin{equation} \label{hh}
\sup_{a\in\mathbb{R}:|a|\le 1} D(P_{n,r,\bp}\|T_aP_{n,r,\bp})
= \max_{k\in\mathbb{Z}:|k|\le n} \sum_{i\in\BZ} p_i\log \frac{p_i}{p_{i+k}},
\end{equation}
then we show that this formula is equal to the objective function in~\eqref{finite_d_opt}. For convenience, we drop the subscripts on $f_{n,r,\bp}$ and $P_{n,r,\bp}$ throughout this proof. We may assume $\bp>\mathbf{0}$, since any vector $\bp$ with some zero coordinate will be infeasible in both optimization problems~\eqref{gj} and~\eqref{finite_d_opt}.

Fix $a\in [-1,1]$. For each $i \in \BZ$, let $\calJ_{n,i}^{\circ} = \left( \frac{i-1/2}{n}, \frac{i+1/2}{n} \right)$ denote the interior of $\calJ_{n,i}$. We start by showing that the function
\begin{equation}
    F_a := f\log \frac{f}{T_{-a}f}
\end{equation}
is integrable, which would allow us to use countable additivity of the Lebesgue integral to split $D(P\|T_{-a}P)$ into a sum of integrals over the $\calJ_{n,i}^{\circ}$. Let $k\in \BZ$ be the unique integer such that $a+\frac{1}{2n} \in \calJ_{n,k}$, and denote $\Delta:=k-an$. From
\begin{equation}
    \frac{k-1/2}{n} \le a +\frac{1}{2n} \le \frac{k+1/2}{n},
\end{equation}
we conclude that $0\le \Delta \le 1$. Consider an integer $i$ and a real $x\in \calJ_{n,i}^\circ$. If $x<(i-1/2+\Delta)/n$, then 
\begin{equation} \label{hd}
    x+a = x + \frac{k-\Delta}{n} < \frac{i+k-1/2}{n} = \frac{(i+k-1)+1/2}{n}
\end{equation}
and, since $\Delta \le 1$,
\begin{equation} \label{he}
    x+a = x + \frac{k-\Delta}{n} > \frac{i-1/2}{n} + \frac{k-1}{n} = \frac{(i+k-1)-1/2}{n}.
\end{equation}
Inequalities~\eqref{hd} and~\eqref{he} together imply that $x+a\in \calJ_{n,i+k-1}^\circ$. Similarly, if $x>(i-1/2+\Delta)/n$ then $x+a\in \calJ_{n,i+k}^\circ$. We may ignore the countably many cases $x=(i-1/2+\Delta)/n$ (as $i$ varies over $\BZ$) for the sake of integrating $F_a$. We conclude that for every $x\in \BR$ such that $nx-\Delta+\frac12$ is not an integer,
\begin{equation}
    F_a(x) = \left\{ \begin{array}{cl}
    np_i \log \frac{p_i}{p_{i+k-1}}, &\text{if } ~ x\in \calJ_{n,i}, ~~ x<\frac{i-1/2+\Delta}{n}, \\
    np_i \log \frac{p_i}{p_{i+k}}, & \text{if } ~ x\in \calJ_{n,i}, ~~ x>\frac{i-1/2+\Delta}{n}.
    \end{array} \right.
\end{equation}
Since $\int_{\BR} |F_a| = \sum_{i\in \BZ} \int_{\calJ_{n,i}^\circ} |F_a|$, we obtain
\begin{equation} \label{hf}
    \int_{\BR} |F_a| = \sum_{i\in \BZ} p_i \left( \Delta \left| \log \frac{p_i}{p_{i+k-1}} \right| + (1-\Delta)\left| \log \frac{p_i}{p_{i+k}} \right| \right). 
\end{equation}
Now, we may conclude that $F_a\in L^1(\BR)$ by comparison with a geometric series. Indeed, we show the convergence of the series 
\begin{equation}
    S_\ell := \sum_{i\in \BZ} p_i \left| \log \frac{p_i}{p_{i+\ell}} \right|
\end{equation}
for each fixed $\ell \in \BZ$. Consider the set of indices
\begin{equation}
    I = \BZ \setminus \{-N-|\ell|,\cdots,N+|\ell|\},
\end{equation}
and note that for each $i\in I$ we have $p_{i+j} = p_N r^{|i+j|-N}$ for both values $j\in \{0,\ell\}$. In particular, for $i\in I$ we have that
\begin{equation}
    \left| \log \frac{p_i}{p_{i+\ell}} \right| = \left| |i|-|i+\ell| \right| \cdot \log\frac{1}{r} \le |\ell| \cdot \log\frac{1}{r}.
\end{equation}
Therefore, we obtain the bound
\begin{equation}
    S_\ell \le \frac{|\ell| p_N \log\frac{1}{r}}{r^N}\cdot \frac{1+r}{1-r} + \sum_{|i|\le N+ |\ell|} p_i \left| \log \frac{p_i}{p_{i+\ell}} \right| < \infty.
\end{equation}
As $S_k$ and $S_{k-1}$ are both finite, we conclude from~\eqref{hf} that $F_a\in L^1(\BR)$. Therefore, by countable additivity,
\begin{equation}
    D(P\|T_{-a}P)=\sum_{i\in \BZ} \int_{\calJ_{n,i}^\circ} F_a,
\end{equation}
i.e.,
\begin{equation} \label{hg}
    D(P\|T_{-a}P) =  \sum_{i\in \BZ} p_i \left( \Delta  \log \frac{p_i}{p_{i+k-1}}  + (1-\Delta)\log \frac{p_i}{p_{i+k}} \right).
\end{equation}

Let $B_\ell$ denote the same sum as $S_\ell$ but without the absolute value sign,
\begin{equation}
    B_\ell := \sum_{i\in \BZ} p_i  \log \frac{p_i}{p_{i+\ell}}.
\end{equation}
Finiteness of the $S_\ell$ yields from~\eqref{hg} that
\begin{equation} \label{hi}
    D(P\|T_{-a}P) = \Delta B_{k-1} + (1-\Delta) B_k.
\end{equation}
Also, the relation we are aiming to prove~\eqref{hh} can be restated as
\begin{equation}
    \sup_{|d|\le 1} D(P\|T_dP) = \max_{|\ell| \le n} B_\ell.
\end{equation}
We deduce from $k=an+\Delta$,  $|a|\le 1$, and $0\le \Delta \le 1$ that we must have $-n\le k \le n+1$. If it holds that $-n+1 \le k \le n$, then what we have shown in~\eqref{hi} implies, in view of $0\le \Delta \le 1$, that
\begin{equation}
    D(P\|T_{-a}P) \le \max_{|\ell|\le n} B_\ell.
\end{equation}
We treat the remaining two extreme cases $k\in \{-n,n+1\}$ separately. First, if $k=-n$ then $\Delta=0$, in which case
\begin{equation}
    D(P\| T_{-a}P) = B_{-n} \le  \max_{|\ell|\le n} B_\ell.
\end{equation}
Second, if $k=n+1$ then $\Delta=1$, in which case
\begin{equation}
    D(P\| T_{-a}P) = B_{n} \le  \max_{|\ell|\le n} B_\ell.
\end{equation}
Combining all cases, we conclude that
\begin{equation} \label{hj}
    \sup_{|d|\le 1} D(P\| T_{d}P)  \le  \max_{|\ell|\le n} B_\ell.
\end{equation}

We establish now that the reverse inequality in~\eqref{hj} also holds. Let $\ell \in \{0,\cdots,n\}$. The shift $a_\ell:=\ell/n$ satisfies $|a_\ell|\le 1$ and $a_\ell + \frac{1}{2n} \in \calJ_{n,\ell}$. Also, $\Delta_\ell := \ell - a_\ell n = 0$. Therefore, we conclude from~\eqref{hi} that
\begin{equation} \label{hn}
    D(P\|T_{-a_\ell}P) = B_\ell.
\end{equation}
This shows that
\begin{equation} \label{hk}
    \sup_{|d|\le 1} D(P\|T_dP) \ge \max_{0\le \ell \le n} B_\ell.
\end{equation}
In addition, consider $\ell\in \{-n,\cdots,-1\}$ and the shift $a_\ell':=\ell/n$. Then, in this case $a_\ell' + \frac{1}{2n} \in \calJ_{n,\ell+1}$. Also, $\Delta_\ell' := (\ell+1)-a_\ell'n = 1$. Thus, by~\eqref{hi}, we have that
\begin{equation} \label{ho}
    D(P\|T_{-a_\ell'}P) = B_{(\ell+1)-1} = B_\ell.
\end{equation}
Therefore,
\begin{equation} \label{hl}
    \sup_{|d|\le 1} D(P\|T_{-d}P) \ge \max_{-n\le \ell \le -1} B_\ell.
\end{equation}
Combining~\eqref{hk} and~\eqref{hl}, we conclude that
\begin{equation} \label{hm}
    \sup_{|d|\le 1} D(P\|T_{-d}P) \ge \max_{|\ell| \le n} B_\ell.
\end{equation}
Inequality~\eqref{hm} together with the reverse inequality~\eqref{hj} yield that the desired equation~\eqref{hh} holds, i.e.,
\begin{equation}
\sup_{|a|\le 1} ~ D(P\|T_aP)
= \max_{|k|\le n} ~ \sum_{i\in\mathbb{Z}} p_i\log \frac{p_i}{p_{i+k}}.
\end{equation}

Next, we show that the expression 
\begin{equation}
    \max_{|k|\le n} ~ \sum_{i\in\mathbb{Z}} p_i\log \frac{p_i}{p_{i+k}}
\end{equation}
reduces to the form given in the statement of the theorem. By construction, $p_i = p_{-i}$ for each $i\in \BZ$. Thus, we have for each $k\in \BZ$
\begin{align}
    B_k &= \sum_{i\in\mathbb{Z}} p_i\log \frac{p_i}{p_{i+k}} = \sum_{j\in\mathbb{Z}} p_{-j}\log \frac{p_{-j}}{p_{-j+k}} = \sum_{j\in\mathbb{Z}} p_{j}\log \frac{p_{j}}{p_{j-k}} = B_{-k}.
\end{align}
Therefore, $B_k = (B_k+B_{-k})/2$ for every $k\in \BZ$. Note that this is a symmetric expression in $k$. As $B_0=0$, the KL-divergence is nonnegative, and $B_k\ge 0$ for every $|k|\le n$ (see~\eqref{hn} and~\eqref{ho}), we conclude that
\begin{equation} \label{hp}
    \sup_{|a|\le 1} ~ D(P\|T_aP) = \max_{1\le k \le n} ~ \frac{1}{2}(B_k+B_{-k}).
\end{equation}
We now rewrite~\eqref{hp} in terms of $p_i$ for only $0\le i\le N$, by taking advantage of~\eqref{p_i_def}. Fix $k\in \{1,\cdots,n\}$. We may write
\begin{equation}
    B_{-k} = \sum_{j\in \BZ} p_j \log \frac{p_j}{p_{j-k}} = \sum_{i\in \BZ} p_{i+k} \log \frac{p_{i+k}}{p_i},
\end{equation}
so
\begin{equation}
    B_k+B_{-k} = \sum_{i\in \BZ} (p_i-p_{i+k}) \log \frac{p_i}{p_{i+k}}.
\end{equation}
We split this sum at the points $-N,N-k,$ and $N$.  For any $k\in\{1,\ldots,n\}$, using the assumption that $n<N$, we may write
\begin{multline}\label{three_terms}
\sum_{i\in\mathbb{Z}} (p_i-p_{i+k})\log\frac{p_i}{p_{i+k}}
=\sum_{i=-N+1}^{N-k-1} (p_{|i|}-p_{|i+k|})\log \frac{p_{|i|}}{p_{|i+k|}}
\\+\sum_{i=N-k}^\infty (p_i-p_{i+k})\log\frac{p_i}{p_{i+k}}
+\sum_{i=-\infty}^{-N} (p_i-p_{i+k})\log\frac{p_i}{p_{i+k}}.
\end{multline}
In fact, the third term in \eqref{three_terms} is identical to the second. This is proved by
\begin{align}
\sum_{i=-\infty}^{-N} (p_i-p_{i+k})\log\frac{p_i}{p_{i+k}}
&=\sum_{i=N}^\infty (p_{-i}-p_{-i+k})\log \frac{p_{-i}}{p_{-i+k}}
\\&=\sum_{i=N}^\infty (p_{i}-p_{i-k})\log \frac{p_{i}}{p_{i-k}}
\\&=\sum_{i=N-k}^\infty (p_{i+k}-p_{i})\log\frac{p_{i+k}}{p_{i}}
\\&=\sum_{i=N-k}^\infty (p_i-p_{i+k})\log\frac{p_i}{p_{i+k}}.
\end{align}
Moreover, we may rewrite this expression as
\begin{align}
\sum_{i=N-k}^\infty &(p_i-p_{i+k})\log\frac{p_i}{p_{i+k}} \nonumber \\
&=\sum_{i=N-k}^{N-1} (p_i-p_Nr^{i+k-N})\log \frac{p_i}{p_Nr^{i+k-N}}+\sum_{i=N}^\infty (p_N r^{i-N}-p_N r^{i+k-N})\log \frac{p_N r^{i-N}}{p_N r^{i+k-N}}
\\
&=\sum_{i=N-k}^{N-1} (p_i-p_Nr^{i+k-N})\log \frac{p_i}{p_Nr^{i+k-N}} + p_N\sum_{i=N}^\infty  r^{i-N} (1- r^{k})\log r^{-k}
\\
&=\sum_{i=N-k}^{N-1} (p_i-p_Nr^{i+k-N})\log \frac{p_i}{p_Nr^{i+k-N}} +p_N\frac{1- r^{k}}{1-r} k\log r^{-1}.
\end{align}
Putting all of the above together shows that~\eqref{hp} is exactly equal to the objective function in \eqref{finite_d_opt}.

Finally, we show that the cost constraint
\begin{equation}
    \BE_P[c] \le C
\end{equation}
is equivalent to the one given in~\eqref{finite_d_opt}. By nonnegativity of $c$, we have that
\begin{align}
    \BE_P[c] &= \int_{\BR} fc = \sum_{i\in \BZ} \int_{\calJ_{n,i}} np_ic = \sum_{i\in \BZ} p_i c_{n,i} = p_0 c_{n,0} + 2 \sum_{i=1}^{N-1} p_ic_{n,i} + 2p_N \sum_{i=N}^\infty c_{n,i} r^{i-N},
\end{align}
and the proof is complete.

\section{Proof of Theorem~\ref{gk}: Optimality of Cactus} \label{ft}

We will use the integration shorthand
\begin{equation}
    \int_{A} f := \int_A f(x) \, dx.
\end{equation}

Define
\begin{equation} \label{op}
    \gamma := \left\{ \begin{array}{cl}
    1/2 & \text{ if } \alpha>1, \\
    \alpha/2 & \text{ otherwise.}
    \end{array} \right.
\end{equation}
Note that $\gamma \in (0,1/2]$ and $\gamma < \alpha$. Define the PDF
\begin{equation} \label{or}
    \psi(x) := \mathrm{exp}\left( - |x|^\gamma \right) \cdot \chi^{-1},
\end{equation}
where
\begin{equation}
    \chi := \int_\BR \mathrm{exp}\left( - |x|^\gamma \right) \, dx
\end{equation}
is the normalization constant. As $\gamma\in(0,1]$, the function $z\mapsto |z|^\gamma$ is subadditive. Hence, for any $x,y\in \BR$ we have the inequality
\begin{equation} \label{os}
    \frac{\psi(x+y)}{\psi(x)} \le \mathrm{exp}\left( |y|^\gamma \right).
\end{equation}

For each $\sigma>0,$ denote the dilated PDF
\begin{equation} \label{mo}
    \psi^\sigma(x):= \frac{1}{\sigma}\psi\left( \frac{x}{\sigma} \right).
\end{equation}
We denote the result of convolving a PDF $q$ with $\psi^\sigma$ by $q_\sigma$,
\begin{equation} \label{mp}
    q_\sigma := q \ast \psi^\sigma.
\end{equation}
For any $a\in \BR$, it is easy to see that
\begin{equation}
    T_a(q_\sigma) = \left( T_a q \right)_\sigma,
\end{equation}
so we denote this common quantity by $T_aq_\sigma.$ 

Due to the length of the proof, we break down some of the initial steps into the following five auxiliary lemmas. The proof resumes in the subsequent subsection.

\subsection{Auxiliary Lemmas}

The first lemma helps reduce the problem to considering only continuous PDFs. Specifically, it shows that a convolution $q_\sigma$ can perform arbitrarily close to how the original PDF $q$ does.

\begin{lemma} \label{mx}
For any PDF $q$ and constant $\eta>0$, there is a constant $\sigma_0\in (0,1)$ such that $\sigma\in (0,\sigma_0]$ implies the inequalities
\begin{align}
    D(q_\sigma \|T_a q_\sigma) &\le D(q\|T_a q), \quad \text{for all } a\in \BR, \label{ml} \\
    \BE_{q_\sigma}[c] &\le \BE_q[c] + \eta. \label{ms}
\end{align}
\end{lemma}
\begin{proof}
First, by the data-processing inequality, for any $a\in\mathbb{R}$ and $\sigma>0$,
\begin{equation}
    D(q_{\sigma}\|T_aq_\sigma)\le D(q\|T_a q).
\end{equation}
Thus,~\eqref{ml} always holds. We may assume that $\BE_q[c]<\infty$, for otherwise~\eqref{ms} trivially holds. Now, we will establish~\eqref{ms} for all small $\sigma$ by proving the limit
\begin{equation} \label{mm}
    \lim_{\sigma \to 0^+} \BE_{q_\sigma}[c] = \BE_{q}[c].
\end{equation}

Let $(\Omega,\calF,P)$ be a probability space and $Z,V:\Omega \to \BR$ be independent random variables with PDFs $q$ and $\psi$, respectively, with respect to $\lambda$, i.e., with $P_Z(B) := P(Z^{-1}(B))$ and $P_V(B) := P(V^{-1}(B))$ we have
\begin{equation}
    \frac{dP_Z}{d\lambda} = q, \qquad \frac{dP_V}{d\lambda} = \psi.
\end{equation}
Then, for any $\sigma>0$, the random variable $Z_\sigma := Z+\sigma V$ has  PDF $q_\sigma$ (see equations~\eqref{op}--\eqref{mp}). Denote integration against $P$ by $\BE$; in particular, 
\begin{equation}
    \BE[f(Z,V)]:=\int_\Omega f(Z(\omega),V(\omega)) \, dP(\omega)
\end{equation}
for any Borel function $f:\BR^2 \to \BR$.

By Slutsky's theorem, we have that $Z_\sigma \to Z$ in distribution. By the continuous mapping theorem, we also have that $c(Z_\sigma) \to c(Z)$ in distribution. Thus, by the Lebesgue-Vitali theorem~\cite[Theorem 4.5.4]{Bogachev_book}, to conclude that~\eqref{mm} holds, it suffices to show uniform integrability of $\{c(Z_\sigma)\}_{0<\sigma \le 1}$, i.e., it suffices to show that
\begin{equation} \label{mq}
    \lim_{K\to \infty} \, \sup_{0<\sigma \le 1}\BE\left[ c(Z_\sigma) \cdot 1_{(K,\infty)}(c(Z_\sigma)) \right] = 0.
\end{equation}
To establish~\eqref{mq}, it suffices to uniformly upper bound the $c(Z_\sigma)$ (for $\sigma \in (0,1]$) by an integrable random variable. To see this, note that if
\begin{equation} \label{oq}
    \sup_{0 < \sigma \le 1} c(Z_\sigma) \le U
\end{equation}
for some random variable $U:\Omega \to \BR$ with $\BE[U] < \infty$, then we have the inequality
\begin{align}
    \sup_{0<\sigma \le 1} \BE\left[ c(Z_\sigma) \cdot 1_{(K,\infty)}(c(Z_\sigma)) \right] &\le \BE\left[ U \cdot 1_{(K,\infty)}(U) \right],
\end{align}
and the limit 
\begin{equation}
    \lim_{K\to \infty} \BE\left[ U \cdot 1_{(K,\infty)}(U) \right] = 0
\end{equation}
follows by absolute continuity of the Lebesgue integral in view of $\BE[U]< \infty$.

Now, we show that a uniform bound as in~\eqref{oq} holds. Recall that for any $(u,v)\in\BR^2$ and $0<s<t$, denoting $\|(u,v)\|_s:=(|u|^s+|v|^s)^{1/s}$, one has from H\"{o}lder's inequality that
\begin{equation}
    \|(u,v)\|_t \le \|(u,v)\|_s \le 2^{\frac{1}{s}-\frac{1}{t}} \|(u,v)\|_t.
\end{equation}
In particular, for any $r>0$, denoting $\ell_{r}:= \max(1,2^{r-1})$, one has that
\begin{equation}
    (|u|+|v|)^r \le \ell_{r} (|u|^r+|v|^r).
\end{equation}
In addition, by the tail-regularity assumption on $c$, there is a constant $\beta_1 > 0$ such that
\begin{equation} \label{mr}
    c(x) \le \beta_1\left( 1 + |x|^\alpha \right)
\end{equation}
for every $x\in \BR$. Then, for any $u,v\in \BR$, we have that
\begin{equation}
    c(u+v) \le \beta_1 \left( 1 + \ell_\alpha \left( |u|^\alpha + |v|^\alpha \right) \right).
\end{equation}
In particular, for every $\sigma \in (0,1]$,
\begin{equation} \label{oi}
    c(Z_\sigma) \le \beta_1 \left( 1 + \ell_\alpha \left( |Z|^\alpha + |V|^\alpha \right) \right)=:U.
\end{equation}
Now, we have that $\BE[|V|^\alpha]<\infty$ by definition of $\psi$. Further, by assumption on $c$, there are $A,\beta_2>0$ such that $|x|>A$ implies 
\begin{equation}
    \beta_2 |x|^\alpha \le c(x).
\end{equation}
Then, as
\begin{align}
    |Z|^\alpha &\le A^\alpha + |Z|^\alpha \cdot 1_{\BR\setminus [-A,A]}(Z) \le A^\alpha + c(Z)/\beta_2
\end{align}
and $\BE[c(Z)]=\BE_q[c] < \infty$ by assumption, we also have that $\BE[|Z|^\alpha]<\infty$. Thus, $\BE[U]<\infty$. Hence, by absolute continuity of the Lebesgue integral, the uniform bound in~\eqref{oi} implies the uniform integrability of the set $\{c(Z_\sigma)\}_{0<\sigma \le 1}$, so~\eqref{mm} follows by the Lebesgue-Vitali theorem, and the proof is complete.
\end{proof}

The following lemma shows that the integrands when computing $D(q_\sigma\|T_aq_\sigma)$ have equi-small tails as $a$ varies over $[-1,1]$. This will allow us to focus on approximating $q_\sigma$ by a cactus distribution only in a bounded interval.
\begin{lemma}\label{KL_integrable}
If the PDF $q$ satisfies
\be\label{KL_q_assumption}
\sup_{|a|\le 1} D(q\|T_a q)<\infty
\ee
then for any $\sigma>0$
\be\label{KL_limit}
\lim_{z\to\infty}\sup_{|a|\le 1} \int_{\bbR\setminus[-z,z]} q_\sigma\left|\log \frac{q_\sigma}{T_a q_\sigma}\right|=0.
\ee
\end{lemma}
\begin{proof}
Assume that $q$ satisfies \eqref{KL_q_assumption}. By the data processing inequality, we also have
\be\label{q_sigma_uniform_bound}
\sup_{|a|\le 1} D(q_\sigma\|T_a q_\sigma)<\infty.
\ee
Suppose, for the sake of contradiction, that \eqref{KL_limit} does not hold. That is, suppose there exists $\eps>0$ where
\be
\limsup_{z\to\infty}\sup_{|a|\le 1} \int_{\bbR\setminus[-z,z]} q_\sigma\left|\log \frac{q_\sigma}{T_a q_\sigma}\right|=\eps.
\ee
This implies that there exists a sequence $\{ (z_n,a_n) \}_{n\in \BN}$, where $z_n \nearrow \infty$ and $\sup_{n\in \BN} |a_n|\le 1$, such that for all $n$
\be \label{ot}
\int_{\bbR\setminus[-z_n,z_n]} q_\sigma\left|\log \frac{q_\sigma}{T_{a_n} q_\sigma}\right|\ge\eps/2.
\ee
Since $[-1,1]$ is a compact set, there exists a convergent subsequence $\{a_{n_k}\}_{k\in \BN}$, say $a_{n_k}\to a$ where $a\in[-1,1]$. Moreover, for any $z>0$, for sufficiently large $k$ we have $z_{n_k}\ge z$, which implies 
\be\label{integral_contradiction}
\limsup_{k\to\infty}\int_{\bbR\setminus[-z,z]} q_\sigma\left|\log \frac{q_\sigma}{T_{a_{n_k}} q_\sigma}\right|\ge\eps/2.
\ee

Recall that $\psi$ is as defined in~\eqref{or} and that, as shown in~\eqref{os}, it satisfies the inequality
\begin{equation} \label{sq}
    \frac{\psi(x+y)}{\psi(x)} \le \mathrm{exp}\left( |y|^\gamma \right)
\end{equation}
for every $x,y\in \BR$. Thus, for any $a,b,z\in\bbR$,
\begin{align}
(T_a q_\sigma)(z)&=q_\sigma(z-a)
\\&=\int_{\bbR} q(x) \frac{1}{\sigma} \psi\left(\frac{z-a-x}{\sigma}\right)dx
\\&\le e^{|a-b|^\gamma/\sigma^\gamma} \int_{\bbR} q(x) \frac{1}{\sigma} \psi\left(\frac{z-b-x}{\sigma}\right)dx
\\&= e^{|a-b|^\gamma/\sigma^\gamma} (T_b q_\sigma)(z). \label{sr}
\end{align}
Thus, for any $a,b\in \BR$, we have the uniform bound
\be
\left\|\log \frac{T_aq_\sigma}{T_b q_\sigma}\right\|_{L^\infty(\BR)} \le \left( \frac{|a-b|}{\sigma} \right)^\gamma.
\ee
Applying this bound to the integral in \eqref{integral_contradiction} gives
\begin{align}
\int_{\bbR\setminus[-z,z]} q_\sigma \cdot \left|\log \frac{q_\sigma}{T_{a_{n_k}} q_\sigma}\right|&=\int_{\bbR\setminus[-z,z]} q_\sigma \cdot \left|\log \frac{q_\sigma}{T_{a} q_\sigma}+\log\frac{T_a q_\sigma}{T_{a_{n_k}}q_\sigma}\right|
\\
&\le \int_{\bbR\setminus[-z,z]} q_\sigma \cdot \left(\left|\log \frac{q_\sigma}{T_{a} q_\sigma}\right|+\left( \frac{|a_{n_k}-a|}{\sigma}\right)^\gamma \right)
\\
&\le \left( \frac{|a_{n_k}-a|}{\sigma}\right)^\gamma + \int_{\bbR\setminus[-z,z]} q_\sigma \cdot \left|\log \frac{q_\sigma}{T_{a} q_\sigma}\right|.
\end{align}
Recalling inequality~\eqref{integral_contradiction} and that $a_{n_k}\to a$ as $k\to\infty$, we have, for any $z>0$,
\be\label{to_be_contradicted}
\int_{\bbR\setminus[-z,z]}q_\sigma\left|\log \frac{q_\sigma}{T_{a} q_\sigma}\right|\ge \eps/2.
\ee

Finally, note that by finiteness of the KL-divergence $D(q_\sigma \| T_a q_\sigma)$ (see~\eqref{q_sigma_uniform_bound}), we also have that
\begin{equation}
    \int_{\BR} q_\sigma\left|\log \frac{q_\sigma}{T_{a} q_\sigma}\right| < \infty.
\end{equation}
Indeed, the function $f(t):= t\log t$ over $(0,\infty)$ is lower bounded by $-1/e$, so dividing the integration region over the two regions where $f$ is positive or negative we obtain
\begin{align}
    \int_{\BR} q_\sigma\left|\log \frac{q_\sigma}{T_{a} q_\sigma}\right| &= \BE_{T_aq_\sigma} \left[ \left| f\circ \frac{q_\sigma}{T_aq_\sigma}  \right| \right] \le D\left( q_\sigma \| T_a q_\sigma \right) + \frac{2}{e} < \infty.
\end{align}
Thus, by the monotone convergence theorem, we must have
\be
\lim_{z\to\infty}\int_{\bbR\setminus[-z,z]}q_\sigma\left|\log \frac{q_\sigma}{T_{a} q_\sigma}\right|=0.
\ee
As this contradicts~\eqref{to_be_contradicted}, the lemma is proved.
\end{proof}

The following lemma gives an $\mathrm{exp}(-O(w^\gamma))$ lower bound on the minimum value of $q_\sigma$ over $[-w,w]$ and on the probability that $Z_\sigma \sim q_\sigma$ exceeds $w$, both as $w\to \infty$.

\begin{lemma} \label{my}
For a PDF $q$ and a constant $\sigma>0$, we have that
\begin{equation} \label{qi}
    \int_{[w,\infty)} q_\sigma = \mathrm{exp}\left(-O(w^{\gamma})\right)
\end{equation}
and
\begin{equation} \label{qj}
    \min_{|x|\le w} q_\sigma(x) = \mathrm{exp}\left( - O(w^\gamma) \right),
\end{equation}
both as $w\to \infty$.
\end{lemma}
\begin{proof}
First, we show that there is a bounded Borel set $B$ with $\lambda(B)>0$ such that 
\begin{equation}
    \mu := \inf_{x\in B} q(x)>0.
\end{equation}
Note that we may remove the boundedness condition on $B$. Indeed, if the Borel set $B$ satisfies $\lambda(B)>0$ and $\inf_{x\in B} q(x)>0$, then the bounded Borel sets $A_m := B\cap [-m,m]$ also satisfy $\lambda(A_m)>0$ and $\inf_{x\in A_m} q(x)>0$ for all large $m$ by continuity of $\lambda$ and the definition of the infimum. Now, to see that such a $B$ exists, consider the Borel sets $B_n:=q^{-1}([1/n,\infty))$ for integers $n\ge 1$. For each $n\ge 1$, we have that $\inf_{x\in B_n} q(x) \ge 1/n$. Suppose, for the sake of contradiction, that $\lambda(B_n)=0$ for each $n$. Then we would have
\begin{align}
    \lambda(q^{-1}((0,\infty))) &= 
    \lambda\left( q^{-1}\left( \bigcup_{n\ge 1} [1/n,\infty) \right) \right) = \lambda\left( \bigcup_{n\ge 1} B_n \right) = 0.
\end{align}
Hence, $q=0$ a.e. However, this would contradict that $q$ is a PDF. Thus, we conclude that $\lambda(B_n)>0$ for some $n$. In short, there must exist a bounded Borel set $B$ with $\lambda(B)>0$ and $\inf_{x\in B} q(x) >0$. Fix such a $B$, and let $x_0>0$ be such that $B\subset [-x_0,x_0]$. 

Recall that we define $q_\sigma = q \ast \psi^\sigma$ (see equations~\eqref{op}--\eqref{mp}). For each $w\in \BR$, Tonelli's theorem implies that
\begin{equation} \label{qf}
    \int_{[w,\infty)} q_\sigma = \int_{\BR} q(x) \int_w^\infty \psi\left( \frac{y-x}{\sigma} \right) \frac{1}{\sigma} \, dy \, dx.
\end{equation}
Performing a change of variable, we have for every $x,w\in \BR$
\begin{equation} \label{qg}
    \int_w^\infty \psi\left( \frac{y-x}{\sigma} \right) \frac{1}{\sigma} \, dy = \int_{[(w-x)/\sigma,\infty)} \psi .
\end{equation}
Further, for any $z\ge 0$, by definition of $\psi$, we have the bound
\begin{equation}
    \int_{[z,\infty)} \psi \ge \int_{[z,z+1]} \psi \ge \mathrm{exp}\left( - \left( z + 1 \right)^{\gamma} \right) \cdot \chi^{-1},
\end{equation}
where $\chi = \int_\BR \mathrm{exp}(-|u|^{\gamma}) \, du$ is the normalization constant for $\psi$. Therefore, whenever $w\ge x$ we have
\begin{equation} \label{qh}
    \int_{[(w-x)/\sigma,\infty)} \psi \ge \mathrm{exp}\left( - \left(\frac{w-x+\sigma}{\sigma} \right)^{\gamma} \right) \cdot \chi^{-1}.
\end{equation}
Now, combining~\eqref{qf} and~\eqref{qg}, nonnegativity of the PDFs $q$ and $\psi$ implies the bound
\begin{equation} \label{mu}
    \int_{[w,\infty)} q_\sigma \ge \int_{B} q(x) \int_{[(w-x)/\sigma,\infty)} \psi(u)  \, du \, dx.
\end{equation}
Since $B\subset [-x_0,x_0]$, we conclude from~\eqref{qh} that for every $w\ge x_0$
\begin{align}
    \int_{[w,\infty)} q_\sigma &\ge \int_B \mu \cdot \mathrm{exp}\left( - \left(\frac{w-x+\sigma}{\sigma} \right)^{\gamma} \right) \cdot \chi^{-1} \, dx \\
    &\ge  \lambda(B) \mu \chi^{-1} \cdot \mathrm{exp}\left( - \left(\frac{w+x_0+\sigma}{\sigma} \right)^{\gamma} \right).
\end{align}
The estimate in~\eqref{qi} follows by taking $w\to \infty$.

Finally, we show that~\eqref{qj} holds. Let $w_0>0$ be such that $\int_{[-w,w]} q \ge 1/2$ for every $w\ge w_0$. Then, for any $w\ge w_0$ and $x\in [-w,w]$,
\begin{align}
    q_\sigma(x) &= \int_\BR q(u) \psi^\sigma(x-u) \, du \\
    &=(\sigma \chi)^{-1} \int_\BR q(u) ~ \mathrm{exp}\left( - |x-u|^\gamma/\sigma^\gamma \right) \, du \\
    &\ge (\sigma \chi)^{-1} \int_{-w}^w q(u) ~ \mathrm{exp}\left( - |x-u|^\gamma/\sigma^\gamma \right) \, du \\
    &\ge (\sigma \chi)^{-1} ~ \mathrm{exp}\left( -(2/\sigma^\gamma) w^\gamma \right) \int_{[-w,w]} q \\
    &\ge (2\sigma \chi)^{-1} ~ \mathrm{exp}\left( -(2/\sigma^\gamma) w^\gamma \right).
\end{align}
The estimate~\eqref{qj} follows by taking $w\to \infty$.
\end{proof}

Conversely, the following lemma gives an upper bound on the tail of any distribution that satisfies the cost constraint.

\begin{lemma} \label{ek}
For any $P\in \sB$, if $\BE_P[c]<\infty$ then 
\begin{equation}
    P(\BR\setminus [-x,x]) = o\left(c(x)^{-1}\right)
\end{equation}
as $x\to \infty.$
\end{lemma}
\begin{proof}
We start by showing that
\begin{equation}
    \lim_{t\to \infty} P\left( \left\{ c > t \right\} \right) \cdot t = 0.
\end{equation}
Denote $f(t):= P\left( \left\{ c > t \right\} \right)$ for $t>0$. Note that $f$ is a decreasing nonnegative function over $(0,\infty)$. Further, $f$ is integrable by nonnegativity of $c$ and by the assumption in the lemma since
\begin{equation}
    \int_0^\infty P\left( \left\{ c > t \right\} \right)  \, dt = \BE_P[c] < \infty.
\end{equation}
We show that these three properties of $f$ yield that $f(t) = o(t^{-1})$.

Suppose, for the sake of contradiction, that there is an $\varepsilon>0$ and an increasing sequence $t_n\nearrow \infty$ of strictly positive numbers such that
\begin{equation} \label{pr}
    f(t_n) \ge \frac{\varepsilon}{t_n}
\end{equation}
for every $n\in \BN$. Since $f$ is decreasing, we infer from~\eqref{pr} that
\begin{equation} \label{ps}
    f(t) \ge \sum_{n\in \BN} \frac{\varepsilon}{t_{n+1}} \cdot 1_{(t_n,t_{n+1}]}(t)
\end{equation}
for every $t>t_1$. Integrating both sides in~\eqref{ps}, integrability of $f$ implies that
\begin{equation} \label{pt}
    \infty > \int_{(t_1,\infty)} f \ge \sum_{n\in \BN} \varepsilon \left( 1 - \frac{t_n}{t_{n+1}} \right).
\end{equation}
By convergence of the series in~\eqref{pt}, we conclude that $t_n/t_{n+1}\sim 1$ as $n\to \infty$. In particular, the constant
\begin{equation}
    \tau := \inf_{n\in \BN} \frac{t_n}{t_{n+1}}
\end{equation}
satisfies $\tau \in (0,1)$. Set $\delta = \varepsilon \tau$, and note that $\delta > 0$. Then, from~\eqref{ps} we obtain
\begin{align}
    f(t)\cdot 1_{(t_1,\infty)}(t) &\ge \sum_{n\in \BN}  \frac{\delta}{\tau \cdot t_{n+1}} \cdot 1_{(t_n,t_{n+1}]}(t) \\
    &\ge \sum_{n\in \BN} \frac{\delta}{t_n} \cdot 1_{(t_n,t_{n+1}]}(t) \\
    &\ge \sum_{n\in \BN} \frac{\delta}{t} \cdot 1_{(t_n,t_{n+1}]}(t) = \frac{\delta}{t} \cdot 1_{(t_1,\infty)}(t). \label{pu}
\end{align}
However,~\eqref{pu} contradicts the integrability of $f$. Thus, we conclude that it must be the case that
\begin{equation} \label{pv}
    f(t) = o\left( t^{-1} \right)
\end{equation}
as $t\to \infty$.

To finish the proof of the lemma, recall by the tail-regularity assumption on $c$ that we have  
\begin{equation}
    \lim_{|x|\to\infty} \frac{c(x)}{|x|^\alpha}=\beta
\end{equation}
for some $\alpha,\beta>0$. Thus there are $0<\beta_1<\beta<\beta_2$ such that for sufficiently large $|x|$, $\beta_1 |x|^\alpha\le c(x)\le \beta_2 |x|^\alpha$. Then, for all large enough $t$, we have that
\begin{align}
    P\left(\left\{ c>t \right\} \right)  &\ge P\left( \BR \setminus \left[ -(t/\beta_1)^{1/\alpha}, (t/\beta_1)^{1/\alpha} \right] \right).
\end{align}
Writing $x=(t/\beta_1)^{1/\alpha}$, we conclude that for all large $x$
\begin{align}
    c(x) P\left( \BR \setminus [-x,x] \right) &\le \beta_2 x^\alpha P\left( \BR \setminus [-x,x] \right) \\
    &= \frac{\beta_2}{\beta_1} t P\left( \BR \setminus \left[ -(t/\beta_1)^{1/\alpha}, (t/\beta_1)^{1/\alpha} \right] \right) \\
    &\le \frac{\beta_2}{\beta_1} t P\left(\left\{ c>t \right\} \right).
\end{align}
Taking $t\to \infty$, we obtain from~\eqref{pv} that
\begin{equation}
    c(x) P\left( \BR \setminus [-x,x] \right) \to 0,
\end{equation}
as desired.
\end{proof}

The final auxiliary lemma gives an upper bound on the tail of the cost constraint incurred by a cactus distribution.

\begin{lemma} \label{ns}
Fix $r\in (0,1)$ and integers $N>n\ge 1$, and set $w=(N-1/2)/n$. Assume that $c(x)\le \beta_1 x^\alpha$ for $x\ge w$. Then, we have the bound
\begin{equation}
    \sum_{i\ge N} c_{n,i} r^{i-N} \le \beta_1 \ell_\alpha \left( \frac{w^\alpha}{1-r} + \frac{2\left( \frac{\alpha}{e} \right)^\alpha \log \frac{1}{r} + \Gamma(\alpha+1)}{r n^\alpha \left( \log \frac{1}{r} \right)^{\alpha +1 }} \right),
\end{equation}
where $\ell_{\alpha}:= \max(1,2^{\alpha-1})$.
\end{lemma}
\begin{proof}
By monotonicity of $c$,
\begin{align}
    \sum_{i\ge N} c_{n,i} r^{i-N} &= \sum_{i\ge N} \int_{(i-1/2)/n}^{(i+1/2)/n} ncr^{i-N} \\
    &\le \sum_{i\ge N} \beta_1 \left( \frac{i+1/2}{n} \right)^\alpha r^{i-N} \\
    &= \beta_1\sum_{i\ge 0} \left( w + \frac{i+1}{n} \right)^\alpha r^i \\
    &\le \beta_1 \ell_\alpha \left( \frac{w^\alpha}{1-r} + \frac{\text{Li}_{-\alpha}(r)}{r n^\alpha} \right),
\end{align}
where
\begin{equation}
    \text{Li}_{-\alpha}(r) := \sum_{k\ge 1} k^\alpha r^k
\end{equation}
is the polylogarithm function. To finish the proof of the lemma, we show next that
\begin{equation}
    \text{Li}_{-\alpha}(r) \le 2\left( \frac{\alpha}{e \log \frac{1}{r}} \right)^\alpha  + \frac{\Gamma(\alpha+1)}{ \left( \log \frac{1}{r} \right)^{\alpha +1 }}.
\end{equation}
Now, consider the function $g:(0,\infty) \to (0,\infty)$ defined by
\begin{equation}
    g(x) := x^\alpha r^x.
\end{equation}
We have that
\begin{equation}
    g'(x) = \left( \alpha +x \log r \right) x^{\alpha-1} r^x.
\end{equation}
Thus, $g$ increases until it reaches a maximum at $x_0=\alpha/\log \frac{1}{r}$ then it decreases. Thus,
\begin{equation}
    \text{Li}_{-\alpha}(r) \le g(\lfloor x_0 \rfloor) + g(\lceil x_0 \rceil) + \int_{(0,\infty)} g.
\end{equation}
We have 
\begin{equation}
    g(\lfloor x_0 \rfloor) + g(\lceil x_0 \rceil) \le 2g(x_0) = 2 \left( \frac{\alpha}{e \log \frac{1}{r}} \right)^\alpha,
\end{equation}
and
\begin{equation}
    \int_{(0,\infty)} g = \frac{\Gamma(\alpha + 1)}{\left( \log \frac{1}{r} \right)^{\alpha + 1}}.
\end{equation}
The proof is thus complete.
\end{proof}

\subsection{Proof of Theorem~\ref{gk}}

By Theorem~\ref{gb}, there is a PDF $q^\star$ that satisfies both
\begin{align}
    \sup_{|a|\le 1} D(q^\star\|T_a q^\star) &= \text{KL}^\star, \label{ok} \\
    \BE_{q^\star}[c] &\le C. \label{ol}
\end{align}
We may assume that $q^\star$ is even; indeed, we may replace $q^\star$ with the even PDF $(q^\star(x)+q^\star(-x))/2$, which satisfies the cost constraint by evenness of $c$, and which also has a better KL-divergence than that of $q^\star$ by joint convexity of the KL-divergence. Fix arbitrary constants $\delta,\eta>0$, and we will find a cactus distribution that attains the KL-divergence~\eqref{ok} to within $\delta$ and the cost~\eqref{ol} to within $\eta$.

By Lemma~\ref{mx}, there is a $\sigma>0$ such that the PDF $q_\sigma^\star$ satisfies the bounds
\begin{align}
    \sup_{|a|\le 1} D(q_\sigma^\star\|T_a q_\sigma^\star) &\le \text{KL}^\star, \label{nb} \\
    \BE_{q_\sigma^\star}[c] &\le C + \frac{\eta}{2}. \label{na}
\end{align}
Throughout the proof, we will denote
\begin{equation} \label{om}
    q := q_\sigma^\star
\end{equation}
for short. Let 
\begin{equation} \label{on}
    Q(B) := \int_B q
\end{equation}
be the probability measure induced by $q$. We will construct a cactus distribution that approximates $q$.

We first note a few properties of $q$. Note that $q$ is an even PDF. Further, it is uniformly continuous, and strictly positive over $\BR$. Thus, $q$ is locally bounded away from zero. For each $z \ge 0$, denote the minimum
\begin{equation} \label{sp}
    \mu_z := \min_{|x|\le z} q(x),
\end{equation}
so $\mu_z > 0$ for every $z$. In addition, $q$ is upper bounded: by Young's inequality, we have that
\begin{align}
    \|q\|_{L^\infty(\BR)} &= \|q^\star \ast \psi^\sigma \|_{L^\infty(\BR)} \le \|q^\star\|_{L^1(\BR)} \cdot \|\psi^\sigma\|_{L^\infty(\BR)} = (\sigma \chi)^{-1}=:M.
\end{align}
In fact, $q$ satisfies a property resembling local $\gamma$-H\"{o}lder continuity. Specifically, as in the proof of Lemma~\ref{KL_integrable} (see~\eqref{sq}--\eqref{sr}), we have that
\begin{equation}
    q(x) \le e^{|x-y|^\gamma/\sigma^\gamma} q(y)
\end{equation}
for every $x,y\in \BR$. Therefore, for some $|t_{x,y}|\le 1$ we have
\begin{align}
    |q(x)-q(y)| &=q(y) \left| e^{t_{x,y}|x-y|^\gamma/\sigma^\gamma}-1 \right| \le \frac{2M}{\sigma^\gamma}|x-y|^\gamma,
\end{align}
where the latter inequality follows whenever $|x-y|\leq \sigma$. In particular, for all $\varepsilon \in (0,2M)$, we have that
\begin{equation} \label{ss}
    |q(x)-q(y)| \le \varepsilon \quad \text{whenever} \quad |x-y|\le \sigma \cdot \left( \frac{\varepsilon}{2M} \right)^{1/\gamma}.
\end{equation}

Before constructing the parameters $(n,N,r)$ of the cactus distribution, we note a fundamental lower bound on $n$. For the cost constraint to be satisfied, we need $c_{n,0}<C$ to hold. Nevertheless, by continuity of $c$, every real number is a Lebesgue point of $c$. In particular, as $0$ is a Lebesgue point of $c$, we obtain
\begin{equation}
    c_{n,0} = \frac{\int_{[-1/(2n),1/(2n)]} c}{1/n} \to c(0) = 0
\end{equation}
as $n\to \infty$. Let $n_{\min}$ be the least positive integer such that
\begin{equation}
    c_{n,0}<C
\end{equation}
for every $n\ge n_{\min}$. Note that $n_{\min}$ depends only on $c$ and $C$.

Now, we choose the integers $n$ and $N$. Denote the constants
\begin{align} 
    \theta_\alpha &:= 4\left( \frac{\alpha}{e} \right)^\alpha + 2\Gamma(\alpha+1)  \label{nv} \\
    \theta_\alpha' &:= (2\theta_\alpha)^{1/\alpha} \\
    \gamma' &:= \frac{\gamma+\alpha}{2} \in (\gamma,\alpha) \label{sy} \\
    \varepsilon_{\min} &:= 2M\cdot \min\left(\frac{2}{\sigma n_{\min}} , \frac{1}{\theta_\alpha'\sigma} \right)^\gamma \\
    z_{\min,0} &:= \left(\log\left( \frac{4}{\sigma}\cdot \left( \frac{2M}{\varepsilon_{\min}} \right)^{1/\gamma} \right) \right)^{1/\gamma'} \\
    z_{\min,1} &:= \left( \frac{\eta}{\delta} \cdot \frac{2e\theta_\alpha'}{ \beta_1 \ell_\alpha} \right)^{1/(\alpha+1)} \\
    z_{\min,2} &:= \left( \frac{2^{\alpha+1}}{\beta_1\ell_\alpha} \right)^{1/(\alpha-\gamma')} \\
    z_{\min} &:= \max\left(  z_{\min,0}, z_{\min,1},z_{\min,2},\frac{\delta}{12M}\right). \label{oe}
\end{align}
Since $q=q_\sigma^\star$ (see~\eqref{om}), Lemma~\ref{KL_integrable} yields the existence of a constant $z_0 > 0$ such that $z\ge z_0$ implies the uniform bound
\begin{equation} \label{nn}
    \sup_{|a|\le 1} \int_{\BR \setminus [-z,z]} q \left| \log \frac{q}{T_a q} \right| \le \frac{\delta}{3}.
\end{equation}
In addition, Lemma~\ref{my} yields the existence of constants $\tau,z_1>0$ such that $z\ge z_1$ implies (see~\eqref{on} and~\eqref{sp})
\begin{equation} \label{sw}
    \min\left( \mu_z, Q([z,\infty)) \right) \ge \mathrm{exp}\left( - \tau z^{\gamma} \right).
\end{equation}
By the tail-regularity assumption on $c$, there are constants $\beta_1,\beta_2,z_2 >0$ such that 
\begin{equation} \label{ln}
    \beta_2 z^\alpha \le c(z) \le \beta_1 z^\alpha
\end{equation}
for every $z\ge z_2$. By Lemma~\ref{ek}, we have that (see~\eqref{on})
\begin{equation}
    \lim_{z\to \infty} Q(\BR\setminus [-z,z])c(z) = 0.
\end{equation}
Let $z_3> 0$ be large enough that $z \ge z_3$ implies 
\begin{equation} \label{lo}
    Q\left( \BR \setminus [-z,z] \right) c(z) \le \frac{\beta_2 }{\beta_1 \ell_\alpha}\cdot \frac{\eta}{6}.
\end{equation}
If $z\ge \max(z_2,z_3)$, then by~\eqref{ln} and~\eqref{lo} we may bound the tail of $Q$ also by
\begin{equation} \label{lp}
    Q\left( \BR \setminus [-z,z] \right) \le \frac{1}{\beta_1\ell_\alpha z^\alpha} \cdot \frac{\eta}{6}.
\end{equation}
Let $z_4>0$ be the smallest number such that both inequalities
\begin{align}
    e^{\tau z^\gamma} &\ge \frac{\delta \beta_1 \ell_\alpha}{2M \eta} \cdot z^\alpha \label{st} \\
    e^{\gamma z^{\gamma'}} &\ge \left(\frac{4}{\sigma} \right)^\gamma \cdot \frac{48M^2}{\delta} \cdot ze^{\tau z^\gamma} \label{su}
\end{align}
hold for all $z\ge z_4$. Fix a rational number 
\begin{equation}
    z > \max(z_{\min},z_0,z_1,z_2,z_3,z_4,2\theta_\alpha')
\end{equation}
that is a ratio of an odd integer by an even integer, and set
\begin{equation}
    w := z+1.
\end{equation}
We choose $z$ (hence also $w$) here to belong in $\BN+\frac{1}{2}$ for simplicity, but we note that any other choice (of denominator) is also valid provided that $w$ is increased so that the subsequent choices in~\eqref{nf} below can be made. Set
\begin{equation} \label{sv}
    \varepsilon := \frac{2^{2\gamma+1}M}{\sigma^\gamma} \cdot e^{-\gamma w^{\gamma'}}.
\end{equation}
Denote 
\begin{equation}
    n_{0} := \frac{2}{\sigma} \cdot \left( \frac{2M}{\varepsilon} \right)^{1/\gamma},
\end{equation}
By the uniform continuity of $q$ shown in~\eqref{ss}, we have that
\begin{equation} \label{ng}
    |q(x)-q(y)| \le \varepsilon \quad \text{whenever} \quad |x-y|\le \frac{2}{n_0}.
\end{equation}
Note that $n_{\min} < n_{0}$ since $\varepsilon < \varepsilon_{\min}$, which in turn follows because $w>z_{\min,0}$. We note also that $\varepsilon < \varepsilon_{\min}$ implies $2\theta_\alpha' < n_0$. Set
\begin{equation}
    n_1 := e^{w^{\gamma'}}.
\end{equation}
By construction, we have that $n_1=2n_0$. Thus, we may choose integers $n\in [n_0,n_1]$ and $N>n$ such that
\begin{equation} \label{nf}
    w = \frac{2N-1}{2n}
\end{equation}

Next, we choose the parameter $r$, thereby completing the cactus distribution construction. Define, for $i \in \{0,\cdots,N-1\}$,
\begin{equation} \label{nh}
    p_i := \inf_{x\in \calJ_{n,i}} \frac{q(x)}{n}.
\end{equation}
By evenness, continuity, and strict positivity of $q$, we have that
\begin{equation} \label{nw}
    p_0 + \sum_{i=1}^{N-1} 2p_i = \int_{[-w,w]} \sum_{|i|\le N-1} np_{|i|} \cdot 1_{\calJ_{n,i}} \le \int_{[-w,w]} q < 1.
\end{equation}
Thus, for any $r\in (0,1),$ setting 
\begin{equation} \label{nx}
    p_N := \frac{1-r}{2}\left( 1 - \left( p_0 + \sum_{i=1}^{N-1} 2p_i \right) \right),
\end{equation}
we infer from~\eqref{nw} that the vector $\bp = (p_0,\cdots,p_N)$ belongs to $(0,1]^{N+1}$, and by construction it satisfies $S_{r,\bp}=1$. We will choose $r$ as
\begin{equation} \label{ny}
    r := 1 - \frac{\theta_\alpha'}{wn},
\end{equation}
and define $p_N$ as in~\eqref{nx} for this choice of $r$.

Therefore, $f_{n,r,\bp}$ is a valid cactus distribution. By uniform continuity of $q$ (see~\eqref{ng}) and by definition of the $p_i$ (see~\eqref{nh}), we have that $f_{n,r,\bp}$ uniformly approximates $q$ from below over $[-w,w]$: for every $x\in [-w,w]$ we have that
\begin{equation} \label{ni}
    0 \le q(x) - f_{n,r,\bp}(x) \le \varepsilon.
\end{equation}
We will deduce from the uniform bound~\eqref{ni} that $f_{n,r,\bp}$ approximates $q$ in the two senses:
\begin{equation} \label{nc}
    \BE_{f_{n,r,\bp}}[c] \le \BE_q[c] + \frac{\eta}{2}
\end{equation}
and
\begin{equation} \label{nd}
    \sup_{|a|\le 1} D(f_{n,r,\bp}\|T_af_{n,r,\bp}) \le \sup_{|a|\le 1} D(q\|T_aq) + \delta.
\end{equation}
Combined with~\eqref{nb}--\eqref{na}, we would conclude from~\eqref{nc}--\eqref{nd} that
\begin{equation} \label{ne}
    \BE_{f_{n,r,\bp}}[c] \le C+\eta
\end{equation}
and
\begin{equation} \label{nm}
    \sup_{|a|\le 1} D(f_{n,r,\bp}\|T_af_{n,r,\bp})  \le \text{KL}^\star + \delta.
\end{equation}

Now, we show that $f_{n,r,\bp}$ satisfies the cost constraint~\eqref{ne}. Since $f_{n,r,\bp}|_{[-w,w]}\le q|_{[-w,w]}$, we have that
\begin{equation} \label{od}
    \BE_{f_{n,r,\bp}}[c\cdot 1_{[-w,w]}] \le \BE_{q}[c] \le C + \frac{\eta}{2}.
\end{equation}
We show next that
\begin{equation} \label{oc}
    \BE_{f_{n,r,\bp}}[c\cdot 1_{\BR \setminus [-w,w]}] \le \frac{\eta}{2}.
\end{equation}
By construction of $f_{n,r,\bp}$, and since $w=(N-1/2)/n$ (see~\eqref{nf}), we have the expression
\begin{equation} \label{nl}
     \BE_{f_{n,r,\bp}}[c\cdot 1_{\BR \setminus [-w,w]}] = 2p_N \sum_{i\ge N} c_{n,i} r^{i-N}.
\end{equation}
We bound the terms $2p_N$ and $\sum_{i\ge N} c_{n,i}r^{i-N}$ separately. By Lemma~\ref{ns}, we have the bound
\begin{equation} \label{nt}
    \sum_{i\ge N} c_{n,i} r^{i-N} \le \beta_1 \ell_\alpha \left( \frac{w^\alpha}{1-r} + \frac{2\left( \frac{\alpha}{e} \right)^\alpha \log \frac{1}{r} + \Gamma(\alpha+1)}{r n^\alpha \left( \log \frac{1}{r} \right)^{\alpha +1 }} \right).
\end{equation}
By definition of $r$ (see~\eqref{ny}), and since $w\ge 1$ and $n\ge 2\theta_\alpha'$, we have that $r\ge 1/2 > 1/e$. Thus, we deduce from~\eqref{nt} that
\begin{equation} \label{nu}
    \sum_{i\ge N} c_{n,i} r^{i-N} \le \beta_1 \ell_\alpha \left( \frac{w^\alpha}{1-r} + \frac{\theta_\alpha}{n^\alpha \left( \log \frac{1}{r} \right)^{\alpha +1 }} \right),
\end{equation}
where $\theta_\alpha$ is as defined in~\eqref{nv}. In addition, we have that (recall that we denote by $P_{n,r,\bp}$ the probability measure associated with $f_{n,r,\bp}$)
\begin{equation}
    \frac{2p_N}{1-r} = P_{n,r,\bp}\left( \BR \setminus [-w,w] \right) = 1 - P_{n,r,\bp}\left(  [-w,w] \right).
\end{equation}
As $f_{n,r,\bp}$ uniformly approximates $q$ from below over $[-w,w]$ to within $\varepsilon$ (see~\eqref{ni}), we have that
\begin{equation}
    P_{n,r,\bp}\left(  [-w,w] \right) \ge Q\left(  [-w,w] \right) - 2\varepsilon w.
\end{equation}
Thus, by the bound on the tail of $Q$ in~\eqref{lp}
\begin{align}
    \frac{2p_N}{1-r} &\le Q\left( \BR \setminus [-w,w] \right) + 2\varepsilon w \le \frac{1}{\beta_1\ell_\alpha w^\alpha} \cdot \frac{\eta}{6} + 2\varepsilon w. \label{lr}
\end{align}
Further, combining inequalities~\eqref{st}--\eqref{su} and using the definition of $\varepsilon$ in~\eqref{sv}, we obtain
\begin{equation} \label{te}
    \varepsilon \le \frac{\eta}{12\beta_1 \ell_\alpha w^{\alpha+1}}.
\end{equation}
Thus, we deduce
\begin{equation} \label{nk}
    2p_N  \le \frac{\eta \cdot (1-r)}{3\beta_1 \ell_\alpha w^\alpha}.
\end{equation}
From the expression in~\eqref{nl}, multiplying inequalities~\eqref{nu} and~\eqref{nk} and noting that $1-r \le \log \frac{1}{r}$, we obtain
\begin{equation} \label{ob}
    \BE_{f_{n,r,\bp}}[c\cdot 1_{\BR \setminus [-w,w]}] \le \frac{\eta}{3}\left( 1 +  \frac{\theta_\alpha}{\left( wn\log \frac{1}{r}\right)^\alpha }  \right).
\end{equation}
By definition of $r$, we have that 
\begin{equation} \label{oa}
\log \frac{1}{r} \ge 1-r = \frac{\theta_\alpha'}{wn}.
\end{equation}
Using inequality~\eqref{oa} in~\eqref{ob}, we obtain
\begin{equation}
    \BE_{f_{n,r,\bp}}[c\cdot 1_{\BR \setminus [-w,w]}] \le \frac{\eta}{3} \cdot \frac{3}{2} = \frac{\eta}{2},
\end{equation}
which is inequality~\eqref{oc}. Combining~\eqref{od}--\eqref{oc}, we deduce~\eqref{ne}, i.e.,
\begin{equation}
    \BE_{f_{n,r,\bp}}[c] \le C+\eta.
\end{equation}

Next, we show that $f_{n,r,\bp}$ satisfies the KL bound~\eqref{nm}. We begin by splitting the integration at the points $\pm z$. By finiteness of the considered KL-divergences, we have for each $|a|\le 1$
\begin{align}
    D(f_{n,r,\bp}\|T_{-a} f_{n,r,\bp})  - D(q \| T_{-a} q)   &\le  \int_{[-z,z]} \left( f_{n,r,\bp} \log \frac{f_{n,r,\bp}}{T_{-a} f_{n,r,\bp}} - q\log \frac{q}{T_{-a}q} \right) \nonumber \\
    & \hspace{1.45cm} + \int_{\BR \setminus [-z,z]}  f_{n,r,\bp} \log \frac{f_{n,r,\bp}}{T_{-a} f_{n,r,\bp}} \nonumber \\
    & \hspace{1.45cm} + \int_{\BR \setminus [-z,z]} q \log \frac{q}{T_{-a} q}. \label{lv}
\end{align}
We already have a uniform bound for the last integral in~\eqref{lv}: since $z\ge z_0$, the estimate in~\eqref{nn} holds and we obtain
\begin{equation} \label{tg}
    \sup_{|a|\le 1} \int_{\BR \setminus [-z,z]} q \log \frac{q}{T_a q}  \le \frac{\delta}{3}.
\end{equation}
We proceed to bounding the first integral in~\eqref{lv} uniformly by
\begin{equation} \label{th}
     \sup_{|a|\le 1} \int_{[-z,z]} \left( f_{n,r,\bp} \log \frac{f_{n,r,\bp}}{T_{-a} f_{n,r,\bp}} - q\log \frac{q}{T_{-a}q} \right) \le \frac{\delta}{3}.
\end{equation}
We do this via deriving an upper bound on the integrand that is uniform in both $a$ and the variable of integration. From $w\ge \delta/(12M)$~\eqref{oe}, $\mu_w\ge e^{-\tau w^\gamma}$~\eqref{sw}, and~\eqref{su}, we have that
\begin{equation} \label{sx}
    \varepsilon \le \frac{\mu_w}{2}\cdot \min\left( 1, \frac{\delta}{12Mw} \right).
\end{equation}
Define the function $g:[-w,w]\to [0,\varepsilon]$ by
\begin{equation}
    g := q - f_{n,r,\bp}.
\end{equation}
That the range of $g$ is contained within $[0,\varepsilon]$ follows since $f_{n,r,\bp}$ approximates $q$ from below uniformly over $[-w,w]$ to within $\varepsilon$. Thus, $z=w-1$ yields
\begin{equation}
    \sup_{|a|\le 1} \| T_ag \|_{L^\infty([-z,z])} \le \varepsilon.
\end{equation}
We note that, over $[-z,z]$, the inequality
\begin{align}
    f_{n,r,\bp} \log \frac{f_{n,r,\bp}}{T_{-a} f_{n,r,\bp}} - q\log \frac{q}{T_{-a}q} \le -q \log\left( 1 - T_{-a} \frac{g}{q} \right) - g \log \left( 1- \frac{g}{q} \right) + g \log \frac{T_{-a}q}{q} \label{lw}
\end{align}
holds; that all the logarithms are well defined follows since $g\le q$ over $[-w,w]$. Indeed, subtracting the left hand side from the right hand side in~\eqref{lw}, we get the function
\begin{equation}
    -q \log \left( 1 - \frac{g}{q} \right) - g \log\left( 1 - T_{-a} \frac{g}{q} \right),
\end{equation}
which is nonnegative over $[-z,z]$ since $g$ is nonnegative over $[-w,w]$. Now, we bound each of the terms in~\eqref{lw}. It is easy to see that for $0\le t \le 1/2$ one has
\begin{equation}
    - \log (1-t) \le 2t.
\end{equation}
Now, we show that $g/q \le 1/2$ over $[-w,w]$. Indeed, this is equivalent to $q\le 2f_{n,r,\bp}$ over $[-w,w]$. But $q-\varepsilon \le f_{n,r,\bp}$ over $[-w,w]$, which implies in view of $\varepsilon \le \mu_w/2 \le q/2$ (over $[-w,w]$) that $q\le 2f_{n,r,\bp}$, as desired. Thus, we obtain that over $[-z,z]$
\begin{equation} \label{no}
    -q \log\left( 1 - T_{-a} \frac{g}{q} \right) \le 2q T_{-a} \frac{g}{q} \le \frac{2M\varepsilon}{\mu_w},
\end{equation}
and
\begin{equation} \label{np}
    -g \log \left( 1- \frac{g}{q} \right) \le \frac{2g^2}{q} \le \frac{2\varepsilon^2}{\mu_w} \le \varepsilon.
\end{equation}
It is also clear that over $[-z,z]$
\begin{equation} \label{nq}
    g \log \frac{T_{-a}q}{q} \le \varepsilon \log \frac{M}{\mu_w} \le \varepsilon \left( \frac{M}{\mu_w} - 1 \right).
\end{equation}
Plugging in inequalities~\eqref{no}--\eqref{nq} into~\eqref{lw}, we obtain the uniform bound
\begin{equation}
    f_{n,r,\bp} \log \frac{f_{n,r,\bp}}{T_{-a} f_{n,r,\bp}} - q\log \frac{q}{T_{-a}q} \le \frac{3M \varepsilon}{\mu_w}
\end{equation}
over $[-z,z]$. Integrating, we deduce
\begin{align}
    \int_{[-z,z]} f_{n,r,\bp} \log \frac{f_{n,r,\bp}}{T_{-a} f_{n,r,\bp}} - q\log \frac{q}{T_{-a}q} &\le \frac{6zM\varepsilon}{\mu_w} < \frac{\delta}{3}, \label{lx}
\end{align}
where the last inequality follows by~\eqref{sx}.

It remains to upper bound the middle integral in~\eqref{lw}, for which we also derive a uniform upper bound
\begin{equation} \label{nr}
    \sup_{|a|\le 1} \int_{\BR \setminus [-z,z]}  f_{n,r,\bp} \log \frac{f_{n,r,\bp}}{T_{-a} f_{n,r,\bp}} \le \frac{\delta}{3}.
\end{equation}
We will further split the integration at the points $\pm(w+1)$. By evenness of $f_{n,r,\bp}$, we have that this integral depends only on $|a|$, i.e., for each $a\in [-1,1]$
\begin{equation}
    \int_{\BR \setminus [-z,z]}  \hspace{-0.4pt} f_{n,r,\bp} \log \frac{f_{n,r,\bp}}{T_{-a} f_{n,r,\bp}} = \int_{\BR \setminus [-z,z]}  \hspace{-0.4pt} f_{n,r,\bp} \log \frac{f_{n,r,\bp}}{T_{a} f_{n,r,\bp}}.
\end{equation}
Thus, it suffices for~\eqref{nr} to show that
\begin{equation}
    \sup_{0<a\le 1} \int_{\BR \setminus [-z,z]}  f_{n,r,\bp} \log \frac{f_{n,r,\bp}}{T_{-a} f_{n,r,\bp}} \le \frac{\delta}{3}.
\end{equation}
Consider first the integral
\begin{equation} \label{lz}
    \int_{\BR \setminus [-(w+1),w+1]} f_{n,r,\bp} \log \frac{f_{n,r,\bp}}{T_{-a} f_{n,r,\bp}}
\end{equation}
for fixed $a\in (0,1]$. From the proof of Theorem~\ref{eh}, we can write the integrand in~\eqref{lz} as follows. Extend the definition of $p_i$ to all $i\in \BZ$ by
\begin{equation}
    p_i := \left\{ \begin{array}{cl} 
    p_{|i|}, & \text{if } -N\le i \le -1, \vspace{2mm} \\
    p_Nr^{|i|-N}, & \text{if } |i|>N.
    \end{array} \right.
\end{equation}
For each $i\in \BZ$, there is an integer $j$ with $|j| \le n$, such that we have 
\begin{equation} \label{ma}
    f_{n,r,\bp} \log \frac{f_{n,r,\bp}}{T_{-a} f_{n,r,\bp}} = np_i \log \frac{p_i}{p_{i+j}}
\end{equation}
over $\calJ_{n,i}$ except possibly at a single point. By definition of $w$, we have that
\begin{equation}
    \BR \setminus [-(w+1),w+1] = \bigcup_{|i|\ge N+n} \calJ_{n,i}.
\end{equation}
Further, if $|i|\le N+n$ and $|j|\le n$, then $|i+j|\ge N$. Hence, from~\eqref{ma} we have that over $\calJ_{n,i}$ with $|i|\ge N+n$
\begin{align}
    f_{n,r,\bp} \log \frac{f_{n,r,\bp}}{T_{-a} f_{n,r,\bp}} &= n p_Nr^{|i|-N} (|i|-|i+j|) \log r \le n^2 p_Nr^{|i|-N}  \log \frac{1}{r}.
\end{align}
Summing over $|i|\ge N+n$, we obtain
\begin{align}
    \int_{\BR \setminus [-(w+1),w+1]} f_{n,r,\bp} \log \frac{f_{n,r,\bp}}{T_{-a} f_{n,r,\bp}} &= \sum_{|i|\ge N+n} \int_{\calJ_{n,i}} f_{n,r,\bp} \log \frac{f_{n,r,\bp}}{T_{-a} f_{n,r,\bp}} \\
    &\le  np_N\log \frac{1}{r} \sum_{|i| \ge n} r^{|i|} = \frac{2np_N r^n \log \frac{1}{r}}{1-r}.
\end{align}
Using the upper bound on $p_N$ in~\eqref{nk}, we obtain that
\begin{equation}
    \int_{\BR \setminus [-(w+1),w+1]} f_{n,r,\bp} \log \frac{f_{n,r,\bp}}{T_{-a} f_{n,r,\bp}} \le \frac{\eta nr^n \log \frac{1}{r}}{3\beta_1 \ell_\alpha w^\alpha}.
\end{equation}
As $1/e\le r \le 1$ and $\log \frac{1}{r} \le \frac{1}{r}-1$, using the definition of $r$ given in~\eqref{ny} and $w\ge z_4$ (see~\eqref{oe}), we have the bound
\begin{equation}
    \frac{\eta nr^n \log \frac{1}{r}}{3\beta_1 \ell_\alpha w^\alpha} \le \frac{e \eta n (1-r)}{3\beta_1 \ell_\alpha w^\alpha} \le \frac{e \eta  \theta_\alpha'}{3 \beta_1 \ell_\alpha w^{\alpha+1}} \le \frac{\delta}{6}.
\end{equation}
Thus, we have shown that
\begin{equation} \label{tf}
    \sup_{a\in (0,1]} \int_{\BR \setminus [-(w+1),w+1]} f_{n,r,\bp} \log \frac{f_{n,r,\bp}}{T_{-a} f_{n,r,\bp}} \le \frac{\delta}{6}.
\end{equation}

The final integral bound we need is the following:
\begin{equation} \label{mf}
    \sup_{0<a\le 1} \int_{w-1<|x|\le w+1} f_{n,r,\bp}(x) \log \frac{f_{n,r,\bp}(x)}{T_{-a} f_{n,r,\bp}(x)} \, dx \le \frac{\delta}{6}.
\end{equation}
By evenness of $f_{n,r,\bp}$, we have that
\begin{align}
    \int_{w-1<|x|\le w+1} f_{n,r,\bp}(x) \log \frac{f_{n,r,\bp}(x)}{T_{-a} f_{n,r,\bp}(x)} \, dx = \int_{(w-1,w+1]} f_{n,r,\bp} \log \frac{f_{n,r,\bp}^2}{(T_{-a}f_{n,r,\bp})\cdot (T_af_{n,r,\bp})}. \label{td}
\end{align}
Consider the function inside the logarithm in the integrand:
\begin{equation}
    \rho(x;a) := \frac{f_{n,r,\bp}(x)^2}{f_{n,r,\bp}(x+a) f_{n,r,\bp}(x-a)}.
\end{equation}
We will prove the uniform upper bound 
\begin{equation}
    \sup_{\substack{x\in (w-1,w+1] \\
    a\in (0,1]}} \, \rho(x;a) \le  \mathrm{exp}\left( 2w^{\gamma'} \right) ,
\end{equation}
where $\gamma' := (\gamma+\alpha)/2 \in (\gamma,\alpha)$ is as defined in~\eqref{sy}. Note that
\begin{equation}
    (w-1,w+1] = \bigcup_{i=N-n}^{N+n} \calJ_{n,i}.
\end{equation}
For each $a \in (0,1]$ and $x\in (w-1,w+1]$, there are integers $N-n \le i \le N+n$ and $0\le j,k \le n$ such that
\begin{equation}
    \rho(x;a) = \frac{p_i^2}{p_{i+j}p_{i-k}}.
\end{equation}
Thus, it suffices to show that $\mathrm{exp}(w^{\gamma'})$ is an upper bound on each of the terms 
\begin{equation}
    \frac{p_i}{p_j}, ~ \frac{p_k}{p_Nr^n}, ~ \frac{p_N}{p_k}, ~ \frac{1}{r^n}
\end{equation}
for $0\le i,j,k \le N-1$ with $|i-j|\le n$. First, for $1/r^n$, denoting $m = nw/ (2\theta_\alpha)^{1/\alpha} \ge 2$, we have the bound
\begin{equation}
    r^n = \left( \left( 1 - \frac{1}{m} \right)^{m} \right)^{(2\theta_\alpha)^{1/\alpha}/w} \ge 4^{-(2\theta_\alpha)^{1/\alpha}/w} \ge \frac{1}{2}.
\end{equation}
Hence,
\begin{equation} \label{sz}
    \frac{1}{r^n} \le 2 \le e^{w^{\gamma'}}.
\end{equation}
For $p_k/p_N$ with $0\le k \le N-1$, we have the bound
\begin{align}
    \frac{p_k}{p_N} &\le \frac{M}{np_N} = \frac{2M/(1-r)}{n\cdot (2p_N/(1-r))} = \frac{2M/(1-r)}{nP_{n,r,\bp}(\BR\setminus [-w,w])} \\
    &\le \frac{2M/(1-r)}{nQ(\BR\setminus [-w,w])} \le \frac{M/(1-r)}{ne^{-\tau w^{\gamma}}} = \frac{Mwe^{\tau w^{\gamma}}}{\theta_\alpha'}.
\end{align}
Hence,
\begin{equation} \label{ta}
    \frac{p_k}{p_Nr^n} \le \frac{2Mwe^{\tau w^{\gamma}}}{\theta_\alpha'} \le e^{w^{\gamma'}},
\end{equation}
where the last inequality follows from~\eqref{su} for all small $\delta$, e.g., for
\begin{equation}
    \delta \le 3\cdot 2^{2\gamma+2}\cdot \theta_{\alpha}' \cdot \chi^{-1}
\end{equation}
(alternatively, we may increase the size of $w$ at the outset). Consider next $p_i/p_j$ for $0\le i,j \le N-1$ with $|i-j|\le n$. By definition of the $p_k$ and uniform continuity of $q$, we have for $0\le k \le N-2$ 
\begin{equation}
    |p_k-p_{k+1}|\le \frac{\varepsilon}{n}.
\end{equation}
By the triangle inequality, we deduce
\begin{equation}
    \left| p_i - p_j \right| \le \frac{|i-j|\varepsilon}{n} \le \varepsilon. 
\end{equation}
Thus,
\begin{equation} \label{tb}
    \frac{p_i}{p_j} \le 1 + \frac{\varepsilon}{p_j} \le 1 + \frac{n\varepsilon}{\mu_w} \le 1 + \frac{n}{2} \le e^{w^{\gamma'}}.
\end{equation}
The last term $p_N/p_k$ can be bounded using~\eqref{nk} to obtain
\begin{align}
    \frac{p_N}{p_k} &\le \frac{\eta \cdot (1-r)/(6\beta_1 \ell_\alpha w^\alpha)}{\mu_w/n} = \frac{\eta \theta_\alpha'}{6\beta_1 \ell_\alpha \mu_w w^{\alpha+1}} \le  \frac{\eta \theta_\alpha'}{6\beta_1 \ell_\alpha } \cdot e^{\tau w^\gamma} \le e^{w^{\gamma'}}, \label{tc}
\end{align}
where the last inequality follows from~\eqref{su} for all small $\eta$, e.g., for
\begin{equation}
    \eta \le 24\beta_1 \ell_{\alpha} \cdot \chi^{-1} \cdot (\theta_\alpha')^{-2}
\end{equation}
(alternatively, we may increase the size of $w$ at the outset). Collecting~\eqref{sz},~\eqref{ta},~\eqref{tb}, and~\eqref{tc}, we obtain the following upper bound on the integral in~\eqref{td}:
\begin{equation}
    P_{n,r,\bp}((w-1,w+1]) \cdot 2w^{\gamma'}.
\end{equation}
Further, 
\begin{align}
    P_{n,r,\bp}((w-1,w+1]) ]&\le P_{n,r,\bp}((w-1,w]) + P_{n,r,\bp}((w,\infty)) \\
    &\le Q((w-1,w]) + \frac12 - P_{n,r,\bp}([0,w]) \\
    &\le Q((w-1,w]) + \frac12 - \left(  Q([0,w]) - \varepsilon w \right) \\
    &= \varepsilon w + Q((z,\infty)) \\
    &\le \varepsilon w + \frac{\eta}{12 \beta_1 \ell_\alpha z^\alpha} \\
    &\le \frac{\eta}{6 \beta_1 \ell_\alpha z^\alpha},
\end{align}
where the last inequality follows by~\eqref{te}. Hence, the integral in~\eqref{td} is upper bounded by
\begin{equation}
    \frac{2^{\alpha}}{3\beta_1\ell_\alpha w^{\alpha - \gamma'}} \cdot \eta \le \frac{\eta}{6},
\end{equation}
where the last inequality follows since $w\ge z_{\min}$ (see~\eqref{oe}). Thus, we have shown that~\eqref{mf} holds, which when combined with~\eqref{tf} gives~\eqref{nr}.

Combining~\eqref{tg},~\eqref{th}, and~\eqref{nr} gives, in view of~\eqref{lv}, the desired inequality~\eqref{nm}:
\begin{equation}
    \sup_{|a|\le 1} D(f_{n,r,\bp}\|T_af_{n,r,\bp})  \le \text{KL}^\star + \delta.
\end{equation}
Recall that we showed in~\eqref{ne} that
\begin{equation}
    \BE_{f_{n,r,\bp}}[c] \le C+\eta.
\end{equation}
To sum up, we make the dependence on $C$ explicit in the optimal values, i.e., we write $\KL^\star(C)$,  $\KL_{n,N,r}^\star(C)$, and $\KL_{\text{Cactus}}^\star(C)$. What we have shown above yields that
\begin{equation}
    \KL_{n,N,r}^\star(C+\eta) \le \KL^\star(C) + \delta.
\end{equation}
Consider the values
\begin{equation}
    \KL_{\textup{Cactus}}^\circ(C) := \inf_{(n,N,r)\in \BN^2\times (0,1)} \KL_{n,N,r}^\star(C),
\end{equation}
so (as defined by~\eqref{sa} in the statement of the theorem) $\KL_{\text{Cactus}}^\star(C) = \lim_{\eta \to 0^+} \KL_{\text{Cactus}}^\circ(C+\eta)$. We conclude that
\begin{equation}
    \KL^\star(C+\eta) \le \KL_{\textup{Cactus}}^\circ(C+\eta) \le \KL^\star(C) + \delta.
\end{equation}
Taking $\delta\to 0^+$, we have
\begin{equation}
    \KL^\star(C+\eta) \le \KL_{\text{Cactus}}^\circ(C+\eta) \le \KL^\star(C).
\end{equation}
Finally, being the infimum of a jointly convex function over a convex set, the function $C\mapsto \KL^\star(C)$ is convex. Since it is also finite, we see that $\KL^\star(C)$ is continuous over $(0,\infty)$. Thus, taking $\eta\to 0^+$, we see that
\begin{equation}
    \KL_{\text{Cactus}}^\star(C) = \KL^\star(C),
\end{equation}
completing the proof of the theorem.

\end{appendices}

\clearpage

\bibliographystyle{IEEEtran}
\bibliography{main}

\begin{thebibliography}{10}
\providecommand{\url}[1]{#1}
\csname url@samestyle\endcsname
\providecommand{\newblock}{\relax}
\providecommand{\bibinfo}[2]{#2}
\providecommand{\BIBentrySTDinterwordspacing}{\spaceskip=0pt\relax}
\providecommand{\BIBentryALTinterwordstretchfactor}{4}
\providecommand{\BIBentryALTinterwordspacing}{\spaceskip=\fontdimen2\font plus
\BIBentryALTinterwordstretchfactor\fontdimen3\font minus
  \fontdimen4\font\relax}
\providecommand{\BIBforeignlanguage}[2]{{%
\expandafter\ifx\csname l@#1\endcsname\relax
\typeout{** WARNING: IEEEtran.bst: No hyphenation pattern has been}%
\typeout{** loaded for the language `#1'. Using the pattern for}%
\typeout{** the default language instead.}%
\else
\language=\csname l@#1\endcsname
\fi
#2}}
\providecommand{\BIBdecl}{\relax}
\BIBdecl

\bibitem{Fisher}
\BIBentryALTinterwordspacing
W.~Alghamdi, S.~Asoodeh, F.~Calmon, O.~Kosut, L.~Sankar, and F.~Wei,
  ``{Schr\"odinger} mechanisms: Optimal differential privacy mechanisms for
  small sensitivity,'' 2022. [Online]. Available:
  \url{https://github.com/WaelAlghamdi/DP-Schrodinger}
\BIBentrySTDinterwordspacing

\bibitem{erlingsson2014rappor}
{\'U}.~Erlingsson, V.~Pihur, and A.~Korolova, ``Rappor: Randomized aggregatable
  privacy-preserving ordinal response,'' in \emph{Proceedings of the 2014 ACM
  SIGSAC conference on computer and communications security}.\hskip 1em plus
  0.5em minus 0.4em\relax ACM, 2014, pp. 1054--1067.

\bibitem{Apple_Privacy}
{Differential privacy team Apple}, ``Learning with privacy at scale,'' 2017.

\bibitem{Facebook020GuidelinesFI}
D.~Kifer, S.~Messing, A.~Roth, A.~Thakurta, and D.~Zhang, ``Guidelines for
  implementing and auditing differentially private systems,'' \emph{ArXiv},
  vol. abs/2002.04049, 2020.

\bibitem{Dwork_Calibration}
C.~Dwork, F.~McSherry, K.~Nissim, and A.~Smith, ``Calibrating noise to
  sensitivity in private data analysis,'' in \emph{Proc. Theory of Cryptography
  (TCC)}, Berlin, Heidelberg, 2006, pp. 265--284.

\bibitem{ghosh2012universally}
A.~Ghosh, T.~Roughgarden, and M.~Sundararajan, ``Universally utility-maximizing
  privacy mechanisms,'' \emph{SIAM Journal on Computing}, vol.~41, no.~6, pp.
  1673--1693, 2012.

\bibitem{Gupte_universally}
M.~Gupte and M.~Sundararajan, ``Universally optimal privacy mechanisms for
  minimax agents,'' in \emph{Proceedings of the Twenty-Ninth ACM
  SIGMOD-SIGACT-SIGART Symposium on Principles of Database Systems}, 2010, p.
  135–146.

\bibitem{geng2015optimal}
Q.~Geng and P.~Viswanath, ``The optimal noise-adding mechanism in differential
  privacy,'' \emph{IEEE Transactions on Information Theory}, vol.~62, no.~2,
  pp. 925--951, 2015.

\bibitem{geng2015staircase}
Q.~Geng, P.~Kairouz, S.~Oh, and P.~Viswanath, ``The staircase mechanism in
  differential privacy,'' \emph{IEEE Journal of Selected Topics in Signal
  Processing}, vol.~9, no.~7, pp. 1176--1184, 2015.

\bibitem{OptimalDP2}
J.~Soria-Comas and J.~Domingo-Ferrer, ``\BIBforeignlanguage{English}{Optimal
  data-independent noise for differential privacy},''
  \emph{\BIBforeignlanguage{English}{Information Sciences}}, vol. 250, no.
  Complete, pp. 200--214, 2013.

\bibitem{Geng_OptimalApproximate}
Q.~Geng and P.~Viswanath, ``Optimal noise adding mechanisms for approximate
  differential privacy,'' \emph{IEEE Transactions on Information Theory},
  vol.~62, no.~2, pp. 952--969, 2016.

\bibitem{Geng_truncatedLaplace}
Q.~Geng, W.~Ding, R.~Guo, and S.~Kumar, ``Tight analysis of privacy and utility
  tradeoff in approximate differential privacy,'' in \emph{Proceedings of the
  Twenty Third International Conference on Artificial Intelligence and
  Statistics}, ser. Proceedings of Machine Learning Research, S.~Chiappa and
  R.~Calandra, Eds., vol. 108, 2020, pp. 89--99.

\bibitem{Geng_Uniform}
\BIBentryALTinterwordspacing
------, ``Optimal noise-adding mechanism in additive differential privacy,'' in
  \emph{Proceedings of the Twenty-Second International Conference on Artificial
  Intelligence and Statistics}, ser. Proceedings of Machine Learning Research,
  K.~Chaudhuri and M.~Sugiyama, Eds., vol.~89.\hskip 1em plus 0.5em minus
  0.4em\relax PMLR, 16--18 Apr 2019, pp. 11--20. [Online]. Available:
  \url{https://proceedings.mlr.press/v89/geng19a.html}
\BIBentrySTDinterwordspacing

\bibitem{dwork2010boosting}
C.~Dwork, G.~N. Rothblum, and S.~Vadhan, ``Boosting and differential privacy,''
  in \emph{51st Annual Symposium on Foundations of Computer Science}.\hskip 1em
  plus 0.5em minus 0.4em\relax IEEE, 2010, pp. 51--60.

\bibitem{Vadhan_Murtagh}
J.~Murtagh and S.~Vadhan, ``The complexity of computing the optimal composition
  of differential privacy,'' in \emph{Proc. Int. Conf. Theory of Cryptography},
  2016, pp. 157--175.

\bibitem{kairouz_Composition}
P.~Kairouz, S.~Oh, and P.~Viswanath, ``The composition theorem for differential
  privacy,'' in \emph{Proceedings of the 32nd International Conference on
  Machine Learning}, F.~Bach and D.~Blei, Eds., vol.~37, 2015, pp. 1376--1385.

\bibitem{abadi2016deep}
M.~Abadi, A.~Chu, I.~Goodfellow, H.~B. McMahan, I.~Mironov, K.~Talwar, and
  L.~Zhang, ``Deep learning with differential privacy,'' in \emph{Proceedings
  of the 2016 ACM SIGSAC conference on computer and communications security},
  2016, pp. 308--318.

\bibitem{asoodeh2021three}
S.~Asoodeh, J.~Liao, F.~P. Calmon, O.~Kosut, and L.~Sankar, ``Three variants of
  differential privacy: Lossless conversion and applications,'' \emph{IEEE
  Journal on Selected Areas in Information Theory}, vol.~2, no.~1, pp.
  208--222, 2021.

\bibitem{Mohammadi_Bucket}
S.~Meiser and E.~Mohammadi, ``Tight on budget? tight bounds for r-fold
  approximate differential privacy,'' in \emph{Proceedings of the 2018 ACM
  SIGSAC Conference on Computer and Communications Security}, ser. CCS '18,
  2018, p. 247–264.

\bibitem{dong2019gaussian}
J.~Dong, A.~Roth, and W.~J. Su, ``Gaussian differential privacy,'' \emph{arXiv
  preprint arXiv:1905.02383}, 2019.

\bibitem{gopi2019numerical}
S.~Gopi, Y.~T. Lee, and L.~Wutschitz, ``Numerical composition of differential
  privacy,'' in \emph{Advances in Neural Information Processing Systems
  (NeurIPS)}, 2021.

\bibitem{koskela2020computing}
A.~Koskela, J.~J{\"a}lk{\"o}, and A.~Honkela, ``Computing tight differential
  privacy guarantees using fft,'' in \emph{International Conference on
  Artificial Intelligence and Statistics}.\hskip 1em plus 0.5em minus
  0.4em\relax PMLR, 2020, pp. 2560--2569.

\bibitem{koskela21a_FFT}
\BIBentryALTinterwordspacing
A.~Koskela, J.~J{\"a}lk{\"o}, L.~Prediger, and A.~Honkela, ``Tight differential
  privacy for discrete-valued mechanisms and for the subsampled gaussian
  mechanism using fft,'' in \emph{Proceedings of The 24th International
  Conference on Artificial Intelligence and Statistics}, ser. Proceedings of
  Machine Learning Research, A.~Banerjee and K.~Fukumizu, Eds., vol. 130.\hskip
  1em plus 0.5em minus 0.4em\relax PMLR, 13--15 Apr 2021, pp. 3358--3366.
  [Online]. Available: \url{https://proceedings.mlr.press/v130/koskela21a.html}
\BIBentrySTDinterwordspacing

\bibitem{zhu2021optimal}
Y.~Zhu, J.~Dong, and Y.-X. Wang, ``Optimal accounting of differential privacy
  via characteristic function,'' \emph{arXiv preprint arXiv:2106.08567}, 2021.

\bibitem{Erhan_book}
E.~\c{C}ınlar, \emph{Probability and Stochastics}.\hskip 1em plus 0.5em minus
  0.4em\relax New York, NY: Springer, 2011.

\bibitem{tfprivacy}
TensorFlow-Privacy tutorial, \url{https://github.com/tensorflow/privacy.git/}.

\bibitem{posner}
E.~Posner, ``Random coding strategies for minimum entropy,'' \emph{IEEE
  Transactions on Information Theory}, vol.~21, no.~4, pp. 388--391, 1975.

\bibitem{Polyanskiy2019}
\BIBentryALTinterwordspacing
Y.~Polyanskiy, ``Lecture notes on information theory,'' 2019. [Online].
  Available:
  \url{http://people.lids.mit.edu/yp/homepage/data/itlectures_v5.pdf}
\BIBentrySTDinterwordspacing

\bibitem{Bogachev_book}
V.~I. Bogachev, \emph{\BIBforeignlanguage{eng}{Measure theory}}.\hskip 1em plus
  0.5em minus 0.4em\relax Berlin: Springer, 2007.

\end{thebibliography}

\end{document}